\documentclass[11pt, a4paper, reqno]{article}
\usepackage{microtype}
\usepackage{amsmath,amsfonts,amsthm}
\usepackage{a4wide}
\usepackage{cases}
\usepackage{tikz}
\usepackage{hyperref}
\hypersetup{colorlinks=true,citecolor=blue,linkcolor=blue,urlcolor=blue}

\DeclareMathOperator{\VCSP}{VCSP}
\DeclareMathOperator{\VCSPp}{VCSP_p}

\DeclareMathOperator{\Feas}{Feas}
\DeclareMathOperator{\Opt}{Opt}
\DeclareMathOperator{\mnrt}{Mn}
\DeclareMathOperator{\mjrt}{Mj}
\DeclareMathOperator{\Pol}{Pol}

\DeclareMathOperator{\ar}{ar}

\newcommand{\mmorp}[1]{\langle #1 \rangle}
\newcommand{\wCloneg}[1]{\operatorname{wClone}(#1)}
\newcommand{\wClonep}[1]{\operatorname{wClone_p}(#1)}
\newcommand{\base}[2]{\tup{e}_{#1}^{#2}}
\newcommand{\negate}[1]{\overline{#1}}
\newcommand{\class}[1]{\textsc{#1}}
\newcommand{\tuple}[1]{\ensuremath{(#1)}}
\newcommand{\tup}[1]{\ensuremath{\mathbf #1}}
\newcommand{\Q}{\mathbb{Q}}
\newcommand{\Qn}{\ensuremath{\mathbb{Q}_{\geq 0}}}
\newcommand{\QInfty}{\overline{\Q}}
\newcommand{\proj}[2]{\ensuremath{\pi^{(#1)}_{#2}}}
\newcommand{\closure}[1]{#1^*}
\newcommand{\clGamma}{\closure{\Gamma}}
\newcommand{\Projection}[2]{\operatorname{Pr}_{#1}(#2)}
\newcommand{\STP}{\mmorp{\sqcap, \sqcup}}
\newcommand{\MJN}{\mmorp{\operatorname{Mj}_1, \operatorname{Mj}_2, \operatorname{Mn}_3}}
\newcommand{\join}{\cdot}
\newcommand{\joinIJ}[2]{\tup{#1}_I \join \tup{#2}_J}
\newcommand{\joinIJp}[2]{\tup{#1}_I \join \tup{#2}_{J'}}
\newcommand{\joinIJpk}[3]{\tup{#1}_I \join \tup{#2}_{J'} \join #3_k}
\newcommand{\PrGamma}{\Projection{i,j}{\Feas(\gamma)}}

\newtheorem{theorem}{Theorem}
\newtheorem{lemma}[theorem]{Lemma}

\theoremstyle{definition}
\newtheorem{definition}[theorem]{Definition}
\newtheorem{example}[theorem]{Example}
\theoremstyle{remark}

\begin{document}
\title{On Planar Valued CSPs\thanks{An extended abstract of part of this work
appeared in the 
\emph{Proceedings of the 41st International Symposium on Mathematical
Foundations of Computer Science (MFCS)}, 2016~\cite{fz16:mfcs}. The authors
were supported by a Royal Society Research Grant. The work was partly done while the
authors were visiting the Simons Institute for the Theory of Computing at UC
Berkeley. Stanislav \v{Z}ivn\'y was
supported by a Royal Society University Research Fellowship.
This project has received funding from the European Research Council (ERC) under
the European Union's Horizon 2020 research and innovation programme (grant
agreement No 714532). The paper reflects only the authors' views and not the
views of the ERC or the European Commission. The European Union is not liable
for any use that may be made of the information contained therein.}}

\author{
Peter Fulla\\
University of Oxford, UK\\
\texttt{peter.fulla@cs.ox.ac.uk}
\and
Stanislav \v{Z}ivn\'{y}\\
University of Oxford, UK\\
\texttt{standa.zivny@cs.ox.ac.uk}
}

\date{}
\maketitle

\begin{abstract}

We study the computational complexity of planar valued constraint satisfaction
problems (VCSPs), which require the \emph{incidence} graph of the instance be planar. 
First, we show that intractable \emph{Boolean} VCSPs have to
be \emph{self-complementary} to be tractable in the planar setting, thus
extending a corresponding result of Dvo\v{r}\'ak and Kupec [ICALP'15] from CSPs
to VCSPs. Second, we give a complete complexity classification of
\emph{conservative} planar VCSPs on arbitrary finite domains. 
In this case planarity does not lead to any new tractable cases and thus our
classification is a sharpening of the classification of conservative VCSPs by
Kolmogorov and \v{Z}ivn\'y [JACM'13].

\end{abstract}

\section{Introduction}
\label{sec:intro}

The valued constraint satisfaction problem (VCSP) is a far-reaching
generalisation of many natural satisfiability, colouring, minimum-cost
homomorphism, and min-cut problems~\cite{Hell08:survey,kz17:survey}. It is
naturally parametrised by its domain and a valued constraint language. A
\emph{domain} $D$ is an arbitrary finite set. A \emph{valued constraint
language}, or just a language, $\Gamma$ is a (usually finite) set of
weighted relations; each weighted relation $\gamma\in\Gamma$ is a function
$\gamma:D^{\ar(\gamma)}\to\QInfty$, where $\ar(\gamma)\in\mathbb{N}^+$ is the
\emph{arity} of $\gamma$ and $\QInfty=\mathbb{Q}\cup\{\infty\}$ is the set of
extended rationals. 

An \emph{instance} $I=(V,D,C)$ of the VCSP on domain $D$
is given by a finite set of $n$ variables $V=\{x_1,\ldots,x_n\}$ and an
objective function $C:D^n\to\QInfty$ expressed as a weighted sum of \emph{valued constraints}
over $V$, i.e.\ $C(x_1,\ldots,x_n)=\sum_{i=1}^q w_i\cdot\gamma_i(\tup{x}_i)$,
where $\gamma_i$ is a weighted relation, $w_i\in\Qn$ is the \emph{weight} and
$\tup{x}_i\in V^{\ar(\gamma_i)}$ the \emph{scope} of the $i$th valued
constraint. (We note that we allow zero weights and for $w_i=0$ we define
$w_i\cdot\infty=\infty$.) Given an instance $I$,
the goal is to find an assignment $s:V\to D$ of domain labels to the variables
that \emph{minimises} $C$. Given a language $\Gamma$, we denote by
$\VCSP(\Gamma)$ the class of all instances $I$ that use only weighted relations
from $\Gamma$ in their objective function.

We now provide a few examples of languages on $D=\{0,1\}$. If $\Gamma_{\sf
nae}=\{\rho\}$ with $\rho(x,y,z)=\infty$ if $x=y=z$ and $\rho(x,y,z)=0$
otherwise, then $\VCSP(\Gamma_{\sf nae})$ captures precisely the
\class{NAE-$3$-Sat} (Not-All-Equal $3$-Satisfiability) problem. To see this,
observe that any instance of $\VCSP(\Gamma_{\sf nae})$ is equivalent to an
instance of \class{NAE-$3$-Sat} over the same variables, each constraint giving
a ternary clause (weights are without effect in this case).
If $\Gamma_{\sf cut}=\{\gamma\}$ with $\gamma(x,y)=1$ if $x=y$ and
$\gamma(x,y)=0$ otherwise, then $\VCSP(\Gamma_{\sf cut})$ captures precisely the weighted
\class{Min-UnCut} problem. If $\Gamma_{\sf is}=\{\rho,\gamma\}$ with
$\rho(x,y)=\infty$ if $x=y=1$ and $\rho(x,y)=0$ otherwise, and $\gamma(x)=1-x$,
then $\VCSP(\Gamma_{\sf is})$ captures precisely the weighted \class{Maximum Independent
Set} problem. Minimisation of bounded-arity submodular functions (or
equivalently, submodular pseudo-Boolean polynomials of bounded degree)
corresponds to $\VCSP(\Gamma_{\sf sub})$ for $\Gamma_{\sf sub}$ consisting of
all weighted relations $\gamma$ that satisfy
$\gamma(\min(\tup{x},\tup{y}))+\gamma(\max(\tup{x},\tup{y}))\leq\gamma(\tup{x})+\gamma(\tup{y})$,
where $\min$ and $\max$ are applied componentwise.

We will be concerned with \emph{exact} solvability of VCSPs. A language
$\Gamma$ is called \emph{tractable} if $\VCSP(\Gamma')$ can be solved (to
optimality) in polynomial time for every finite subset $\Gamma'\subseteq\Gamma$,
and $\Gamma$ is called \emph{intractable} if $\VCSP(\Gamma')$ is NP-hard for
some finite $\Gamma'\subseteq\Gamma$. For instance, $\Gamma_{\sf sub}$ is
tractable~\cite{Cohen06:complexitysoft} whereas $\Gamma_{\sf nae}$, $\Gamma_{\sf
cut}$, $\Gamma_{\sf is}$ are intractable~\cite{Garey79:intractability}.

\subsection{Contribution}

Languages on a two-element domain are called \emph{Boolean}. The
complexity of Boolean valued constraint languages is well understood and eight
tractable cases have been identified~\cite{Cohen06:complexitysoft}. Suppose that
a Boolean language $\Gamma$ is intractable. We are interested in restrictions that can be
imposed on input instances of $\VCSP(\Gamma)$ that make the problem
tractable. A natural way is to restrict the \emph{incidence graph} of the
instance (the precise definition is given in Section~\ref{sec:prelim}). In this
paper we initiate the study of the \emph{planar} variant of the VCSP.

We denote by $\VCSPp(\Gamma)$ the class of instances $I$ of $\VCSP(\Gamma)$ with
planar incidence graph (with an additional requirement that leads to a finer
classification, as discussed in detail in Section~\ref{sec:prelim}). Language $\Gamma$ is called
\emph{planarly-tractable} if $\VCSPp(\Gamma')$ can be solved (to optimality) in
polynomial time for every finite subset $\Gamma'\subseteq\Gamma$, and it
is called \emph{planarly-intractable} if $\VCSPp(\Gamma')$ is NP-hard for some
finite $\Gamma'\subseteq\Gamma$. For instance, while $\Gamma_{\sf nae}$,
$\Gamma_{\sf cut}$, and $\Gamma_{\sf is}$ are
intractable, it is known that $\Gamma_{\sf nae}$
and $\Gamma_{\sf cut}$ are
planarly-tractable~\cite{Moret88:sigact,Hadlock75:sicomp} whereas $\Gamma_{\sf
is}$ is planarly-intractable~\cite{Garey77:siam}.
The problem of classifying all intractable languages as planarly-tractable and
planarly-intractable is challenging and open even for Boolean valued
constraint languages.

A Boolean valued constraint language $\Gamma$ is called
\emph{self-complementary} if every $\gamma\in\Gamma$ satisfies
$\gamma(\tup{x})=\gamma(\negate{\tup{x}})$ for every $\tup{x}\in
D^{\ar(\gamma)}$, where $\negate{\tup{x}}=\tuple{1-x_1,\ldots,1-x_{\ar(\gamma)}}$
for $\tup{x}=\tuple{x_1,\ldots,x_{\ar(\gamma)}}$. As our first contribution, we
show in Section~\ref{sec:Bool} that intractable Boolean valued constraint languages that are \emph{not}
self-complementary are planarly-intractable. We prove this by carefully
constructing planar NP-hardness gadgets for any intractable Boolean valued
constraint language that is not self-complementary, relying on the fact that all
tractable Boolean valued constraint languages are
known~\cite{Cohen06:complexitysoft}. Our result subsumes the analogous
result obtained for $\{0,\infty\}$-valued languages~\cite{Dvorak15:icalp}. We
remark that focusing on Boolean languages is natural as it avoids a number of
difficulties intrinsic to the planar setting. Let $\Gamma_{\sf
col}=\{\gamma\}$ with $\gamma(x,y)=0$ if $x\neq y$ and $\gamma(x,y)=\infty$ otherwise.
Then $\Gamma_{\sf col}$ on domain $D$ with $|D|=3$ is planarly intractable (since
$\VCSPp(\Gamma_{\sf col})$ captures precisely the \class{$3$-Colouring} problem on
planar graphs)~\cite{Garey79:intractability} but is tractable on $D$ with
$|D|=4$ for highly nontrivial reasons, namely the Four Colour Theorem.

A valued constraint language $\Gamma$ on $D$ is called \emph{conservative} if
$\Gamma$ contains all $\{0,1\}$-valued unary weighted relations. The complexity
of conservative valued constraint languages is well understood: a complete
complexity classification has been obtained in~\cite{kz13:jacm}, with a recent
simplification of both the algorithmic and the hardness part~\cite{tz15:icalp,tz17:sicomp}.
As our second contribution, we give a complete complexity classification of
conservative valued constraint languages on arbitrary finite domains with
respect to planar-tractability. In particular, we show that every intractable
conservative valued constraint language is also planarly-intractable. Hence
there are no new tractable cases in the conservative planar setting. This may
seem unsurprising but the proof is not trivial.
We remark that conservative (V)CSPs constitute a large and important fragment of
CSPs~\cite{Bulatov11:conservative} and VCSPs~\cite{kz13:jacm}. In fact, in practice most (V)CSPs are conservative~\cite{Rossi06:handbook}.

Note that for Boolean valued constraint languages that are conservative the
claim follows immediately from our first result: any intractable Boolean
language containing both $\gamma_0(x)=x$ and $\gamma_1(x)=1-x$ (guaranteed by
the conservativity assumption) is not self-complementary, and thus is
planarly-intractable. This shows that $\Gamma=\Gamma_{\sf
cut}\cup\{\gamma_0,\gamma_1\}$ is intractable, a result originally obtained
in~\cite{Barahona82:max-cut} since $\VCSPp(\Gamma)$ captures precisely the
planar \class{Min-UnCut} problem with unary weights. (In fact, the same argument
shows that both $\Gamma_{\sf cut}\cup\{\gamma_0\}$ and $\Gamma_{\sf
cut}\cup\{\gamma_1\}$ are planarly-intractable.)

As it is common in the world of CSPs, dealing with non-Boolean domains
is considerably more difficult than the case of Boolean domains. For
valued constraint languages we have a Galois connection with certain algebraic
objects~\cite{cccjz13:sicomp,fz16:toct} but no Galois connection is known
for valued constraint languages in the planar setting. Moreover, it is
unclear how to use the recent relatively simple proof of the
complexity classification of conservative valued constraint
languages~\cite{tz15:icalp} and make it work in the planar setting since the
proof depends on linear programming duality. (This is related to the lack of a Galois
connection in the planar setting. In particular, \cite[Lemma~2]{tz15:icalp},
which relates (non-planar) expressibility and operator $\Opt$, is proved via LP
duality, and it is unclear how to prove it in the planar setting.)

Our approach is to follow the original
proof of the classification of conservative valued constraint
languages~\cite{kz13:jacm}. In order to adapt the proof for the planar
setting, we significantly simplify it and generalise necessary parts.
Details on proof differences as well as challenges that
we needed to overcome to make the proof work are outlined in
Section~\ref{sec:cons}. We believe that our proof techniques, and in particular the
now simplified and generalised technique from~\cite{kz13:jacm}, will be useful
in future work on planar (V)CSPs.

\subsection{Related work}

VCSPs with $\{0,\infty\}$-valued weighted relations are just (ordinary) decision
CSPs~\cite{Feder98:monotone}. There has been a lot of work on decision CSPs,
see~\cite{Carbonnel16:constraints} for a recent survey. Most results have been
obtained for CSPs parametrised by a constraint language,
see~\cite{Barto14:survey} for a recent survey. Some of the algebraic methods
developed for CSPs~\cite{Bulatov05:classifying} have been extended to
VCSPs~\cite{cccjz13:sicomp,tz15:sidma,fz16:toct,Kozik15:icalp} and successfully
used in classifying various fragments of
VCSPs~\cite{hkp14:sicomp,ktz15:sicomp,tz16:jacm,Kolmogorov15:focs,tz15:icalp}.
However, it is unclear how to use algebraic methods for instance-restricted
classes of VCSPs (sometimes called
\emph{hybrid}~\cite{Carbonnel16:constraints}), even though there are some recent
investigations in this direction~\cite{Kolmogorov15:isaac,Takhanov15:arxiv}.

Following~\cite{Dvorak15:icalp}, we define planar VCSPs by requiring the \emph{incidence} graph be
planar. We note that an alternative option when structurally restricting classes
of (V)CSPs is to consider the \emph{Gaifman} graph, as was done for
CSPs~\cite{Grohe07:jacm}, counting CSPs~\cite{Dalmau04:side}, special cases of
VCSPs~\cite{Farnqvist07:isaac} and also in the setting of parametrised
counting~\cite{Meeks16:dam}. However, we believe that the incidence graph is the
more natural option for the planarity requirement since restricting the Gaifman graph would
exclude (V)CSPs with, for instance, any constraint of arity at least 5.

Planar restrictions have been studied for Boolean (decision)
CSPs~\cite{Dvorak15:icalp,Kazda17:soda}, for Boolean symmetric counting CSPs with
real~\cite{Cai10:focs-planar} and complex~\cite{Guo13:icalp} weights, and also
for Boolean CSPs with respect to polynomial-time approximation
schemes~\cite{Khanna96:stoc,Creignouetal:siam01}.

\section{Preliminaries}
\label{sec:prelim}

\subsection{Planar VCSPs}

Let $I$ be a VCSP instance with variables $V$ and valued constraints $S$. The
\emph{incidence graph} of $I$ is the bipartite multigraph with vertex set $S\cup
V$ and edges $\tuple{\gamma,x_i}$ for every $\gamma(x_1,\ldots,x_{\ar(\gamma)})
\in S$ and $1\leq i\leq \ar(\gamma)$.

We are interested in VCSP instances with \emph{planar} incidence graphs.
Following~\cite{Dvorak15:icalp}, we additionally require the order of edges
around constraint vertices in the plane drawing of the incidence graph
respect the order of arguments of the corresponding constraint. Note that
the variant without this additional restriction can be easily modelled by
replacing each weighted relation $\gamma$ in a language by all weighted relations obtained
from $\gamma$ by permuting the order of its inputs. Hence, this choice leads to
a finer classification.

Following~\cite{Dvorak15:icalp}, rather than working with the incidence graph,
we equivalently define the problem in terms of a related plane graph where
variables correspond to vertices and valued constraints to faces. We note that
our graphs are allowed to have loops, possibly several at a single vertex, and
parallel edges.

For a connected plane graph $G$, we denote by $F(G)$ the set of its faces. For
any face $f \in F(G)$, let $b(f)$ denote a closed walk bounding $f$,
enumerated in the clockwise order around $f$.

\begin{definition}\label{def:planeVCSP}
A \emph{plane $\VCSP$ instance} $\tuple{I,G,\phi}$ is given by a $\VCSP$ instance
$I$ with variables $V$ and objective function $C$ with $q$ valued constraints,
a connected plane graph $G$ over vertices $V$,
and an injective mapping $\phi : \{ 1, \dots, q \} \to F(G)$ such that for every
valued constraint $\gamma_i(x_1, x_2, \dots, x_{\ar(\gamma_i)})$ it holds $b(\phi(i)) = x_1
x_2 \dots{} x_{\ar(\gamma_i)} x_1$.
\end{definition}

\begin{example}\label{ex:inst}
Let $V=\{x_1,x_2,x_3,x_4\}$ and $C(x_1,x_2,x_3,x_4)=2\cdot\gamma_1(x_1)+0\cdot\gamma_2(x_2,x_3,x_1)+\gamma_3(x_3,x_2)+\frac{5}{3}\cdot\gamma_4(x_3,x_4)$.
The (non-planar drawing of the planar) incidence graph of this instance is depicted in Figure~\ref{fig:inst}(a).
The plane graph of the instance from Definition~\ref{def:planeVCSP} is depicted
in Figure~\ref{fig:inst}(b).
\end{example}

\begin{figure}[th]
    \begin{center}
    \begin{tikzpicture}[scale=1,node distance = 1.5cm]
    \tikzstyle{vertex}=[fill=black, draw=black, circle, inner sep=2pt]
    \tikzstyle{dist}  =[fill=white, draw=black, circle, inner sep=2pt]

    \begin{scope}[shift={(-1.0,-0.5)}] 
        \node[vertex] (x1) at (0,3) [label=180:$x_1$] {};
        \node[vertex] (x2) at (0,2) [label=180:$x_2$] {};
        \node[vertex] (x3) at (0,1) [label=180:$x_3$] {};
        \node[vertex] (x4) at (0,0) [label=180:$x_4$] {};
        
        \node[vertex] (g1) at (2.5,3) [label=0:$\gamma_1$] {};
        \node[vertex] (g2) at (2.5,2) [label=0:$\gamma_2$] {};
        \node[vertex] (g3) at (2.5,1) [label=0:$\gamma_3$] {};
        \node[vertex] (g4) at (2.5,0) [label=0:$\gamma_4$] {};

        \draw (g1) -- (x1); \draw (g2) -- (x1); \draw (g2) -- (x2); \draw (g2) -- (x3); \draw (g3) -- (x2); \draw (g3) -- (x3); \draw (g4) -- (x3); \draw (g4) -- (x4);

        \node at (1.25,-1) {a)};
    \end{scope}

    \begin{scope}[shift={(8.0,-0.5)}] 
        \node[vertex] (x1) at (1.25,2.0) [label=180:$x_1$] {};
        \node[vertex] (x2) at (2.5,0.5) [label={[xshift=+0.1cm]$x_2$}] {};
        \node[vertex] (x3) at (0,0.5) [label={[xshift=-0.1cm]$x_3$}] {};
        \node[vertex] (x4) at (-2.5,0.5) [label=$x_4$] {};

        \draw (x1) -- (x2); \draw (x1) -- (x3); 
        \draw (x2) -- (x3);
        \draw (x1).. controls (0,3.5) and (2.5,3.5) .. (x1); 
        \draw (x2).. controls (1.25,-0.5) .. (x3); 
        \draw (x3) -- (x4);
        \draw (x3).. controls (-1.25,-0.5) .. (x4); 

        \node[] (g1) at (1.25,2.60) [] {$\gamma_1$};
        \node[] (g2) at (1.25,1.05) [] {$\gamma_2$};
        \node[] (g3) at (1.25,0.1) [] {$\gamma_3$};
        \node[] (g4) at (-1.25,0.1) [] {$\gamma_4$};

        \node at (0,-1) {b)};
    \end{scope}
    \end{tikzpicture}
    \end{center}
\caption{Graphs from Example~\ref{ex:inst}.}
\label{fig:inst}
\end{figure}
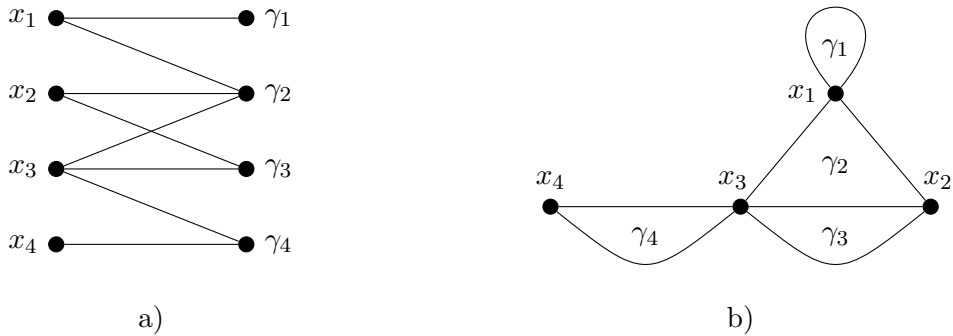

We note that the definition of a \emph{planar} $\VCSP$ instance, in which case
the graph $G$ and mapping $\phi$ are not given, is equivalent to
Definition~\ref{def:planeVCSP}. This is because, as mentioned
in~\cite{Dvorak15:icalp}, checking whether a VCSP instance  $I$ has a planar
representation, and if so then finding $\tuple{I,G,\phi}$, can be done in
polynomial time~\cite{Hopcroft74:planarity}. For simplicity of presentation, we will
assume that graph $G$ and mapping $\phi$ are given.

We denote by $\VCSPp(\Gamma)$ the class of plane $\VCSP$ instances over the language $\Gamma$.

\subsection{Planar Weighted Relational Clones}

In this section, we define planar weighted relational clones, which are closures
of valued constraint languages that do not change the tractability of
corresponding planar VCSPs.

We define \emph{relations} as a special case of weighted relations (also called
\emph{crisp}) with range $\{0, \infty\}$, where value $0$ is assigned to tuples
that are elements of the relation in the conventional sense. For a weighted
relation $\gamma:D^r\to\QInfty$, we denote by
$\Feas(\gamma)=\{\tup{x}\in D^r~|~\gamma(\tup{x})<\infty\}$ the underlying
\emph{feasibility relation}, and by $\Opt(\gamma)=\{\tup{x}\in
\Feas(\gamma)~|~\gamma(\tup{x})\leq \gamma(\tup{y}) \mbox{ for every
}\tup{y}\in D^r\}$ the relation of minimal-value (or \emph{optimal}) tuples. We
also write $\Feas(\gamma)=0\cdot\gamma$ and see the $\Feas$ operator as scaling
a weighted relation by zero, where we define $0\cdot\infty=\infty$.

An assignment $s:V\to D$ for a VCSP instance $(V,D,C)$ with
$V=\{x_1,\ldots,x_n\}$ is called \emph{feasible} if $C(s(x_1),\ldots,s(x_n))<\infty$.

\begin{definition}
Let $\tuple{I,G,\phi}$ be a plane $\VCSP$ instance such that $\phi$ does not map any $i$ to the
outer face $f_o$ of $G$, and let $\tup{v} = (v_1, \dots, v_r)$ be an $r$-tuple
of variables from $V$ such that $b(f_o) = v_r v_{r-1} \dots v_1 v_r$. 
We denote by $\pi_\tup{v}(I)$ the $r$-ary weighted relation mapping any $\tup{x}
\in D^r$ to the minimum objective value obtained by feasible assignments $s$ of
$I$ with $s(\tup{v}) = \tup{x}$, or $\infty$ if no such feasible assignment
exists.

An $r$-ary weighted relation $\gamma$ is \emph{planarly expressible} from a
valued constraint language $\Gamma$ if there exists a plane instance $I$ over
$\Gamma$ and an $r$-tuple $\tup{v}$ of its variables such that $\pi_\tup{v}(I) =
\gamma$.
\end{definition}

\begin{example}\label{ex:inst2}
Let $V=\{x_1,x_3,x_3,z\}$, $D=\{0,1\}$, and $C(x_1,x_3,x_3,z)=\gamma(x_1,z)+\gamma(x_2,z)+\gamma(x_3,z)$ be
a plane VCSP instance $(I,G,\phi)$ depicted in Figure~\ref{fig:inst2}, where $\gamma$ is the
binary ``cut'' weighted relation from Section~\ref{sec:intro}; i.e.,
$\gamma(x,y)=1$ if $x=y$ and $\gamma(x,y)=0$ otherwise. Then
$\rho=\pi_{(x_1,x_2,x_3)}(I)$ is a ternary weighted relation planarly
expressible from $\{\gamma\}$, where $\rho(x,y,z)=0$ if $x=y=z$ and
$\rho(x,y,z)=1$ otherwise.
\end{example}

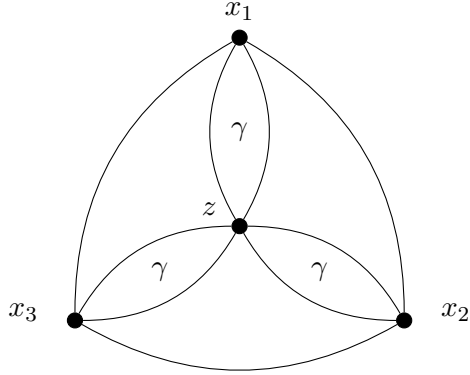
\begin{figure}[h]
    \begin{center}
    \begin{tikzpicture}[scale=1,node distance = 1.5cm]
    \tikzstyle{vertex}=[fill=black, draw=black, circle, inner sep=2pt]
    \tikzstyle{dist}  =[fill=white, draw=black, circle, inner sep=2pt]

        \node[vertex] (z) at (0,0) [label={[xshift=-0.4cm,yshift=-0.1cm]$z$}] {};
        \node[vertex] (x1) at (0,2.5) [label=90:$x_1$] {};
        \node[vertex] (x2) at (2.165,-1.25) [label={[right,xshift=+0.35cm]$x_2$}] {};
        \node[vertex] (x3) at (-2.165,-1.25) [label={[left,xshift=-0.35cm]$x_3$}] {};

        \draw (x1) to [bend right] (z);
        \draw (x1) to [bend left] (z);
        \draw (z) to [bend left] (x2);
        \draw (z) to [bend right] (x2);
        \draw (z) to [bend left] (x3);
        \draw (z) to [bend right] (x3);

        \draw (x1) to [bend right] (x3);
        \draw (x1) to [bend left] (x2);
        \draw (x2) to [bend left] (x3);

        \node[] (r1) at (0,1.25) [] {$\gamma$};
        \node[] (r2) at (1.05,-0.6) [] {$\gamma$};
        \node[] (r2) at (-1.05,-0.6) [] {$\gamma$};

    \end{tikzpicture}
    \end{center}
\caption{Instance from Example~\ref{ex:inst2}.}
\label{fig:inst2}
\end{figure}

To see that planar expressibility is a proper restriction of (unrestricted) 
expressibility~\cite{Cohen06:complexitysoft}, consider relations $\rho_{=} = \{(0,0), (1,1)\}$
and $\rho_\text{cross} = \{(0,0,0,0), (0,1,0,1), (1,0,1,0), (1,1,1,1)\}$ on
domain $D = \{0, 1\}$. Relation $\rho_\text{cross}$ is expressible from the
binary equality relation $\rho_{=}$, because $\rho_\text{cross}(x_1,x_2,x_3,x_4)
= \rho_{=}(x_1,x_3) + \rho_{=}(x_2,x_4)$. However, it is not planarly
expressible. This can be proved unconditionally but here we give a simpler
argument assuming P $\neq$ NP:

Relation $\rho_{=}$ can be included in any valued constraint language without
affecting its complexity (see Lemma~\ref{lmEqInClone} and
Theorem~\ref{thmpClone} below). On the other hand, relation $\rho_\text{cross}$
enables bypassing the planarity restriction; languages from which
$\rho_\text{cross}$ is planarly expressible have the same complexity in the
planar setting as in general \cite{Dvorak15:icalp}. Consequently, if
$\rho_\text{cross}$ is planarly expressible from $\rho_{=}$ then (say) the
\class{NAE-$3$-Sat} problem on general instances can be solved in polynomial
time.

\begin{definition}
A \emph{planar weighted relational clone} is a non-empty set of weighted
relations over the same domain that is closed under planar expressibility,
scaling by non-negative rational constants, addition of rational constants,
and operator $\Opt$. We will denote the smallest planar weighted relational clone
containing a valued constraint language $\Gamma$ by $\wClonep{\Gamma}$.
\end{definition}

An analogous notion of weighted relational clones closed under \emph{general}
(i.e.\ not necessarily planar) expressibility~\cite{cccjz13:sicomp,fz16:toct}
has been used to study the complexity of VCSPs.

\begin{lemma}
\label{lmEqInClone}
For any domain $D$ and language $\Gamma$ on $D$, the binary equality relation
$\rho_{=}$ on $D$ belongs to $\wClonep{\Gamma}$.
\end{lemma}

\begin{proof}
Relation $\rho_{=}$ is planarly expressible by a plane instance consisting of a
single variable $x$ with two self-loops, and $\tup{v} = (x, x)$.
\end{proof}

\begin{theorem}
\label{thmpClone}
For any valued constraint language $\Gamma$, $\Gamma$ is planarly-tractable if,
and only if, $\wClonep{\Gamma}$ is planarly-tractable, and $\Gamma$ is
planarly-intractable if, and only if, $\wClonep{\Gamma}$ is planarly-intractable.
\end{theorem}

\begin{proof}
We show that $\VCSPp(\wClonep{\Gamma})$ is polynomial-time reducible to
$\VCSPp(\Gamma)$. Given an instance $I$ over $\wClonep{\Gamma}$, we replace in
it all weighted relations planarly expressible from $\Gamma$ by their plane
instances. Scaling, which includes $\Feas$, can be achieved by adjusting the weights of the valued constraints.
Adding a constant to a weighted relation affects the value of every feasible
assignment by the same amount, and therefore can be ignored.

Relation $\Opt(\gamma)$ can be simulated by scaling $\gamma$ by a sufficiently
large constant. Let $W$ equal an upper bound on the maximum objective value of a
feasible assignment of $I$. Without loss of generality, we may assume that no
weighted relation of $I$ assigns a negative value and that the smallest value
assigned by $\gamma$ is $0$. Let $d$ equal the second smallest value assigned by
$\gamma$. We replace $\Opt(\gamma)$ with $(W/d + 1) \cdot \gamma$, so that any
assignment of $I$ that would incur an infinite value from $\Opt(\gamma)$ has now
objective value exceeding $W$.
\end{proof}

We now define a few operations on weighted relations that will occur frequently
throughout the paper. As shown in the lemma below, these operations are planarly
expressible.

\begin{definition}
\label{defWRelOperations}
Let $\gamma$ be an $r$-ary weighted relation on $D$.
A \emph{domain restriction} of $\gamma$ to $D'\subseteq D$ at coordinate $i$ is
the $r$-ary weighted relation defined as
$\gamma'(x_1,\ldots,x_r)=\gamma(x_1,\ldots,x_r)$ if $x_i\in D'$ and
$\gamma'(x_1,\ldots,x_r)=\infty$ otherwise.
A \emph{pinning} of $\gamma$ to $a\in D$ at coordinate $i$ is
the $(r-1)$-ary weighted relation defined as
$\gamma'(x_1,\ldots,x_{i-1},x_{i+1},\ldots,x_r)=
\gamma(x_1,\ldots,x_{i-1},a,x_{i+1},\ldots,x_r)$.
Finally, a \emph{minimisation} of $\gamma$ at coordinate $i$ is the $(r-1)$-ary weighted
relation defined as
$\gamma'(x_1,\ldots,x_{i-1},x_{i+1},\ldots,x_r) = \min_{x_i\in
D}\gamma(x_1,\ldots,x_r)$.

A binary weighted relation $\gamma$ is a \emph{join} of two binary weighted
relations $\gamma_1$ and $\gamma_2$ if it can be written as
$\gamma(x,y)=\min_{z\in D}(\gamma_1(u_1,v_1)+\gamma_2(u_2,v_2))$ where
$\{u_1,v_1\}=\{x,z\}$, $\{u_2,v_2\}=\{y,z\}$.
\end{definition}

\begin{lemma}\label{lem:planarunary}
Let us denote by $\rho_{D'}$ the unary relation corresponding to a subdomain $D'
\subseteq D$ (i.e.\ $\rho_{D'}(x) = 0$ if $x \in D'$ and $\rho_{D'}(x) = \infty$
otherwise).

For any language $\Gamma$, $\wClonep{\Gamma}$ is closed under addition of unary
weighted relations to weighted relations of arbitrary arity, minimisation, and
join. If $\rho_{D'} \in \wClonep{\Gamma}$, it is closed under domain restriction
to $D' \subseteq D$. If $\rho_{\{a\}} \in \wClonep{\Gamma}$, it is closed under
pinning to $a \in D$.
\end{lemma}
\begin{proof}
A unary weighted relation $\gamma$ imposed on variable $x_i$ can be planarly
expressed by adding a parallel edge $x_i - x_{i+1}$ and a self-loop at $x_i$
hidden in the just formed face. Minimisation over $x_i$ can be achieved by
adding an edge in the outer face between vertices $x_{i-1}$ and $x_{i+1}$, thus
hiding vertex $x_i$. A join $\gamma(x, y)$ can be achieved by adding two edges
between $x$ and $y$ to hide $z$ from the outer face (similarly as in
Figure~\ref{fig:inst2}). Domain restriction is planarly expressible by imposing
unary relation $\rho_{D'}$ on variable $x_i$; pinning can be expressed by domain
restriction to $\{a\}$ and subsequent minimisation at coordinate $i$.
\end{proof}

Proving results for conservative languages in Section~\ref{sec:cons}, we will
need only a limited subset of $\wClonep{\Gamma}$ which is defined as follows.

\begin{definition}
For any valued constraint language $\Gamma$ on $D$, we define
$\closure{\Gamma}$ to be the smallest set containing $\Gamma$, all unary
weighted relations and the binary equality relation on $D$, and closed under
operators $\Feas$ and $\Opt$, addition of unary weighted relations to weighted
relations of arbitrary arity, minimisation, and join.
\end{definition}

Set $\clGamma$ is also closed under domain restriction and
pinning, as these operations can be achieved by adding unary weighted relations
and minimisation.

Note that for conservative languages we have $\clGamma \subseteq
\wClonep{\Gamma}$, as any unary weighted relation can be obtained from the set
of all $\{0, 1\}$-valued unary weighted relations by addition of unary weighted
relations, scaling, addition of constants, and operator $\Opt$.
By Theorem~\ref{thmpClone}, $\clGamma$ has the same complexity as $\Gamma$.

Lemma~\ref{lmSwapping} will be useful for proving results about both Boolean and
conservative valued constraint languages. Before its statement, we need to
define 2-decomposable relations and introduce some notation.

\begin{definition}
Let $\rho$ be an $r$-ary relation. For any $i, j \in \{1, \dots, r\}$, we will
denote by $\Projection{i,j}{\rho}$ the projection of $\rho$ on coordinates $i$
and $j$, i.e.\ the binary relation defined as
\begin{equation}
(a_i, a_j) \in \Projection{i,j}{\rho} \iff
(\exists \tup{x} \in \rho)~ x_i = a_i \wedge x_j = a_j \,.
\end{equation}
Relation $\rho$ is \emph{2-decomposable} if
\begin{equation}
\tup{x} \in \rho \iff
\bigwedge_{1\leq i,j\leq r} (x_i, x_j) \in \Projection{i,j}{\rho} \,.
\end{equation}
\end{definition}

Note that all unary and binary relations are 2-decomposable.

For any $r$-tuple $\tup{z}$, we denote its $i$th component by $z_i$. Let $I
\subseteq \{ 1, \dots, r \}$ be a subset of coordinates, we denote by
$\tup{z}_I$ the projection of $\tup{z}$ onto $I$. For any partition of
coordinates $I, J \subseteq \{ 1, \dots, r \}$, we then write $\join$ for the
inverse operation, i.e.\ $\joinIJ{z}{z} = \tup{z}$.

\begin{lemma}
\label{lmSwapping}
Let $\gamma$ be an $r$-ary weighted relation and $I, J \subseteq \{1, \dots,
r\}$ a partition of its coordinates. If $\tup{x}, \tup{y} \in \Feas(\gamma)$ and
\begin{equation}
\label{eqSwappingIneq}
\gamma(\tup{x}) + \gamma(\tup{y}) <
\gamma(\joinIJ{x}{y}) + \gamma(\joinIJ{y}{x}) \,,
\end{equation}
then there exist coordinates $i \in I, j \in J$ and a binary weighted relation
$\gamma_{i,j} \in \closure{\{\gamma\}}$ such that $(x_i, x_j), (y_i, y_j) \in
\Feas(\gamma_{i,j})$ and
\begin{equation}
\gamma_{i,j}(x_i, x_j) + \gamma_{i,j}(y_i, y_j) <
\gamma_{i,j}(x_i, y_j) + \gamma_{i,j}(y_i, x_j) \,.
\end{equation}
Moreover, if every relation in $\closure{\{\gamma\}}$ is 2-decomposable, then
$\joinIJ{x}{y} \in \Feas(\gamma)$ implies $(x_i, y_j) \in \Feas(\gamma_{i,j})$
and $\joinIJ{y}{x} \in \Feas(\gamma)$ implies $(y_i, x_j) \in
\Feas(\gamma_{i,j})$.
\end{lemma}

\begin{proof}
We prove the lemma by induction on the arity of $\gamma$. If $|I| = 0$, $|J| =
0$, or $|I| = |J| = 1$, the claim holds trivially. Otherwise we may without loss
of generality assume that $|J| \geq 2$. Let $k \in J$ be an arbitrary coordinate
and define $J' = J \setminus \{k\}$. We extend our notation $\join$ to $I, J',
\{k\}$ as a finer partition of $\{1, \dots, r\}$, and write for instance
$\tup{x}$ as $\joinIJpk{x}{x}{x}$.

We first consider the case when $\joinIJpk{x}{y}{x}, \joinIJpk{y}{x}{y}
\not\in \Feas(\gamma)$. We restrict the domain at coordinate $k$ to $\{x_k,
y_k\}$ and minimise over it to obtain an $(r-1)$-ary weighted relation $\gamma'$
with coordinates partition $I, J'$. It holds $\gamma'(\joinIJp{x}{x}) \leq
\gamma(\tup{x}), \gamma'(\joinIJp{y}{y}) \leq \gamma(\tup{y}),
\gamma'(\joinIJp{x}{y}) = \gamma(\joinIJ{x}{y}), \gamma'(\joinIJp{y}{x}) =
\gamma(\joinIJ{y}{x})$, and the claim follows directly from the induction
hypothesis for $\gamma'$.

We may now assume without loss of generality that $\joinIJpk{y}{x}{y} \in
\Feas(\gamma)$. If
\begin{equation}
\label{eqYXyIneq}
\gamma(\joinIJpk{x}{x}{x}) + \gamma(\joinIJpk{y}{x}{y}) <
\gamma(\joinIJpk{x}{x}{y}) + \gamma(\joinIJpk{y}{x}{x}) \,,
\end{equation}
we pin $\gamma$ at every coordinate $j' \in J'$ to its respective label $x_{j'}$
to obtain a weighted relation $\gamma'$ with coordinates partition $I, \{k\}$.
The claim then follows from the induction hypothesis for $\gamma'$. Note that
$\joinIJ{x}{y} \in \Feas(\gamma)$ implies $(x_i, y_k) \in
\Projection{i,k}{\Feas(\gamma)}$ for all $i \in I$; together with $(x_{j'}, y_k)
\in \Projection{j',k}{\Feas(\gamma)}, (x_i, x_{j'}) \in
\Projection{i,j'}{\Feas(\gamma)}$ for all $i \in I, j' \in J'$ (as
$\joinIJpk{y}{x}{y}, \tup{x} \in \Feas(\gamma)$) this implies
$\joinIJpk{x}{x}{y} \in \Feas(\gamma)$ if $\Feas(\gamma)$ is 2-decomposable.

If \eqref{eqYXyIneq} does not hold, we have $\joinIJpk{x}{x}{y} \in
\Feas(\gamma)$, and therefore
\begin{equation}
\label{eqXXyIneq}
\gamma(\joinIJpk{x}{x}{y}) + \gamma(\joinIJpk{y}{y}{y}) <
\gamma(\joinIJpk{x}{y}{y}) + \gamma(\joinIJpk{y}{x}{y}) \,,
\end{equation}
otherwise the sum of negated \eqref{eqYXyIneq} and \eqref{eqXXyIneq} would
contradict \eqref{eqSwappingIneq}. We resolve this case analogously to the
previous one, this time pinning $\gamma$ at coordinate $k$ to $y_k$.
\end{proof}

\subsection{Algebraic Properties}

We apply a $k$-ary operation $f:D^k\to D$ to $k$ $r$-tuples componentwise;
that is, if $\tup{x}^1=\tuple{x^1_1,\ldots,x^1_r},\tup{x}^2=\tuple{x^2_1,\ldots,x^2_r},\ldots,\tup{x}^k=\tuple{x^k_1,\ldots,x^k_r}$,
then
\[
f(\tup{x}^1,\ldots,\tup{x}^k)\ =\ 
\tuple{f(x^1_1,x^2_1,\ldots,x^k_1), f(x^1_2,x^2_2,\ldots,x^k_2), \ldots,
f(x^1_r,x^2_r,\ldots,x^k_r)}\,.
\]

The following notion is at the heart of the algebraic approach to decision
CSPs~\cite{Bulatov05:classifying}. 

\begin{definition}
Let $\gamma$ be a weighted relation on $D$. 
A $k$-ary operation
$f:D^k\to D$ is a \emph{polymorphism} of $\gamma$ (and $\gamma$ is
\emph{invariant under} or \emph{admits} $f$) if, for every
$\tup{x}^1,\ldots,\tup{x}^k\in\Feas(\gamma)$, we have
$f(\tup{x}^1,\ldots,\tup{x}^k)\in\Feas(\gamma)$. We say that $f$ is a
polymorphism of a language $\Gamma$ if it is a polymorphism
of every $\gamma\in\Gamma$. We denote by $\Pol(\Gamma)$ the set of all
polymorphisms of $\Gamma$.
\end{definition}

A $k$-ary \emph{projection} is an operation of the form
$\proj{k}{i}(x_1,\ldots,x_k)=x_i$ for some $1\leq i\leq k$.
Projections are (trivial) polymorphisms of all valued constraint languages.

The following notion, which involves a collection of $k$ $k$-ary polymorphisms,
played an important role in the complexity
classification of Boolean valued constraint
languages~\cite{Cohen06:complexitysoft}.

\begin{definition}\label{def:mult}
Let $\gamma$ be a weighted relation on $D$. 
A list $\mmorp{f_1,\ldots,f_k}$ of $k$-ary polymorphisms of 
$\gamma$ is a $k$-ary \emph{multimorphism} of $\gamma$ (and $\gamma$ \emph{admits}
$\mmorp{f_1,\ldots,f_k}$) if, for every
$\tup{x}^1,\ldots,\tup{x}^k\in\Feas(\gamma)$, we have
\begin{equation}
\label{eqMMorpIneq}
\sum_{i=1}^k \gamma(f_i(\tup{x}^1,\ldots,\tup{x}^k))\ \leq\ \sum_{i=1}^k
\gamma(\tup{x}^i)\,.
\end{equation}
We say that $\mmorp{f_1,\ldots,f_k}$ is a multimorphism of a language $\Gamma$ if it is a multimorphism of every $\gamma\in\Gamma$. 
\end{definition}

It is known that weighted relational clones preserve polymorphisms and
multimorphisms~\cite{cccjz13:sicomp} and thus planar weighted relational clones
do as well.

\begin{example}
The class of submodular functions on
$D=\{0,1\}$~\cite{Schrijver03:CombOpt} can be defined as the valued constraint
language $\Gamma_{\sf sub}$ that admits $\mmorp{\min,\max}$ as a multimorphism;
that is, for every $\gamma\in\Gamma_{\sf sub}$, we have
$\gamma(\min(\tup{x},\tup{y}))+\gamma(\max(\tup{x},\tup{y}))\leq\gamma(\tup{x})+\gamma(\tup{y})$.
\end{example}

A ternary operation $f:D^3\to D$ is called a \emph{majority} operation if
$f(x,x,y)=f(x,y,x)=f(y,x,x)=x$ for all $x,y\in D$, and a \emph{minority}
operation if $f(x,x,y)=f(x,y,x)=f(y,x,x)=y$ for all $x,y\in D$.

\section{Boolean Valued CSPs}
\label{sec:Bool}

In this section, we will consider only languages on a Boolean domain $D=\{0,1\}$. 
Our first result is that self-complementarity is necessary for
planar-tractability of intractable Boolean languages.

\begin{theorem}\label{thm:main1}
Let $\Gamma$ be a Boolean valued constraint language that is intractable. If $\Gamma$ is not
self-complementary then it is planarly-intractable.
\end{theorem}

We start with some notation for important operations on $D$. For any $a \in D$,
$c_a$ is the constant unary operation such that $c_a(x) = a$ for all $x \in D$.
Operation $\neg$ is the unary negation, i.e.\ $\neg(0) = 1$ and $\neg(1) = 0$.
Binary operation $\min$ ($\max$) is the minimum (maximum) operation with respect
to the order $0 < 1$. Ternary operation $\mnrt$ ($\mjrt$) is the unique minority
(majority) operation on $D$.

Next we define some useful relations. For any $a \in D$, we denote by $\rho_a$
the unary constant relation $\{ (a) \}$. Relation $\rho_{\neq}$ is the binary disequality
relation, i.e.\ $\rho_{\neq} = \{ (0, 1), (1, 0) \}$. Ternary relation
$\rho_\text{1-in-3}$ corresponds to the \class{$1$-in-$3$ Positive $3$-Sat} problem,
i.e.\ $\rho_\text{1-in-3} = \{ (0, 0, 1), (0, 1, 0), (1, 0, 0) \}$. Weighted
relations $\gamma_0, \gamma_1, \gamma_{\neq}$ are defined as soft-constraint
variants of $\rho_0, \rho_1, \rho_{\neq}$ assigning value $0$ to allowed tuples
and $1$ to disallowed tuples.

Note that $\Gamma$ is self-complementary if, and only if, $\Gamma$ admits
multimorphism $\mmorp{\neg}$. The proof of Theorem~\ref{thm:main1} is based on
Lemmas~\ref{lmConstants}--\ref{lmMnrtMjrt} proved below.

We will need the following definition and an easy lemma.

\begin{definition}
Let $\gamma$ be an $r$-ary weighted relation and $i \in \{1, \dots, r\}$. The
\emph{$=$-restriction} of $\gamma$ at $i$ is the $r$-ary weighted relation
$\gamma'$ such that $\gamma'(\tup{x}) = \gamma(\tup{x})$ if $x_i = x_{i+1}$
(where $x_{r+1} = x_1$) and $\gamma'(\tup{x}) = \infty$ otherwise. The
\emph{$\neq$-restriction} of $\gamma$ at $i$ is the $r$-ary weighted relation
$\gamma'$ such that $\gamma'(\tup{x}) = \gamma(\tup{x})$ if $x_i \neq x_{i+1}$
and $\gamma'(\tup{x}) = \infty$ otherwise.

We will denote by $\oplus$ the addition modulo $2$ operation on $\{0, 1\}$ and
its extension to tuples. Let $\tup{0}^r$ ($\tup{1}^r$) be the zero (one)
$r$-tuple. The \emph{negation} of an $r$-tuple $\tup{x}$ is $\negate{\tup{x}} =
\tup{x} \oplus \tup{1}^r$. Let $\base{i}{r}$ be the $r$-tuple with a one at
coordinate $i$ and zeros elsewhere. The \emph{twist} of $\gamma$ at $i$ is the
$r$-ary weighted relation $\gamma'$ defined as $\gamma'(\tup{x}) =
\gamma(\tup{x} \oplus \base{i}{r})$.
\end{definition}

In other words, a twist switches roles of labels 0 and 1 at a single coordinate.

\begin{example}
Let $\rho$ be the ternary ``not-all-equal'' relation from
Section~\ref{sec:intro}; i.e., $\rho(x,y,z)=\infty$ if $x=y=z$ and
$\rho(x,y,z)=0$ otherwise. The twist of $\rho$ at the first coordinate is the
ternary relation $\rho'$ defined by $\rho'(x,y,z)=\infty$ if $x=0$ and $y=z=1$,
or $x=1$ and $y=z=0$; in all other cases $\rho'(x,y,z)=0$.
\end{example}

\begin{lemma}
Let $\Gamma$ be a valued constraint language and $\gamma \in \wClonep{\Gamma}$ a weighted
relation. Then
\begin{itemize}
\item all $=$-restrictions of $\gamma$ belong to
$\wClonep{\Gamma}$,
\item if $\rho_{\neq} \in \wClonep{\Gamma}$, all $\neq$-restrictions and twists
of $\gamma$ belong to $\wClonep{\Gamma}$,
\item if $\rho_0, \rho_1 \in \wClonep{\Gamma}$, all pinnings of $\gamma$ belong
to $\wClonep{\Gamma}$.
\end{itemize}
\end{lemma}

\begin{proof}
Both $=$-restriction and $\neq$-restriction are planarly expressible by adding a
parallel edge between vertices $x_i$, $x_{i+1}$ and imposing on them the binary
equality or disequality relation respectively.
To implement a twist, we introduce a new
variable $x_i'$ in the outer face, connect it with $x_i$ by two parallel edges,
impose the binary disequality relation on $x_i$ and $x_i'$, and hide vertex
$x_i$ by adding edges $x_{i-1} - x_i'$ and $x_{i+1} - x_i'$.
Pinnings belong to $\wClonep{\Gamma}$ by Lemma~\ref{lem:planarunary}.
\end{proof}

\begin{lemma}
\label{lmConstants}
Let $\Gamma$ be a valued constraint language that admits neither of the multimorphisms
$\mmorp{c_0}$, $\mmorp{c_1}$. Then $\rho_0, \rho_1 \in \wClonep{\Gamma}$ or
$\rho_{\neq} \in \wClonep{\Gamma}$.
\end{lemma}

\begin{proof}
If $\Gamma$ does not admit $\mmorp{c_0}$, it contains a weighted relation
assigning to the zero tuple a value larger than the optimum. Applying $\Opt$, we
have that $\wClonep{\Gamma}$ contains a \emph{relation} that is not invariant
under $c_0$. We denote by $\rho$ such a relation of minimum arity and by $r$ its
arity. Relation $\rho$ is non-empty, but $\tup{0}^r \not\in \rho$. If $r = 1$,
then $\rho = \rho_1 \in \wClonep{\Gamma}$.

Otherwise, $\base{i}{r} \in \rho$ for all $i$, because the minimisation of
$\rho$ over coordinate $i$ produces a non-empty relation invariant under $c_0$
(by the choice of $\rho$) and hence
containing $\tup{0}^{r-1}$. If $r \geq 3$, the $=$-restriction of $\rho$ at
coordinate $2$ followed by the minimisation results in an $(r-1)$-ary relation
$\rho'$ with $\base{1}{r-1} \in \rho'$ and $\tup{0}^{r-1} \not\in \rho'$, which
contradicts the choice of $\rho$. Therefore, $r = 2$. If $(1, 1) \in \rho$, we
would again get a contradiction by applying the $=$-restriction and minimisation
at coordinate $1$. Hence we have $\rho = \rho_{\neq} \in \wClonep{\Gamma}$.

By the analogous argument for multimorphism $\mmorp{c_1}$ we get $\rho_0 \in
\wClonep{\Gamma}$ or $\rho_{\neq} \in \wClonep{\Gamma}$.
\end{proof}

\begin{lemma}
\label{lmMinMax}
Let $\Gamma$ be a valued constraint language that admits neither of the multimorphisms
$\mmorp{\min, \min}$, $\mmorp{\max, \max}$, $\mmorp{\min, \max}$. If $\rho_0,
\rho_1 \in \wClonep{\Gamma}$, then $\rho_{\neq} \in \wClonep{\Gamma}$.
\end{lemma}

\begin{proof}
If $\min \not\in \Pol(\wClonep{\Gamma})$, we choose a minimum-arity relation
$\rho_{\vee}' \in \wClonep{\Gamma}$ that is not invariant under $\min$;
its arity $r$ is at least $2$. Let $\tup{x}, \tup{y} \in \rho_{\vee}'$ be
$r$-tuples such that $\min(\tup{x}, \tup{y}) \not\in \rho_{\vee}'$. Tuples
$\tup{x}, \tup{y}$ differ at every coordinate, otherwise we would obtain a
contradiction with the choice of $\rho_{\vee}'$ by taking a pinning instead. Therefore, $\min(\tup{x},
\tup{y}) = \tup{0}^r \not\in \rho_{\vee}'$ and, by the same argument as in Lemma~\ref{lmConstants},
we have $\base{i}{r} \in \rho_{\vee}'$ for all $i$. But then $r = 2$, otherwise we could
take as $\tup{x}, \tup{y}$ tuples $\base{2}{r}, \base{3}{r}$ which agree at
the first coordinate, and obtain a smaller counterexample by pinning. Hence we have $\rho_{\neq} \subseteq \rho_{\vee}'
\subseteq \rho_{\neq} \cup \{ (1,1) \}$.

If $\min \in \Pol(\wClonep{\Gamma})$, then we choose a minimum-arity weighted
relation $\gamma \in \wClonep{\Gamma}$ that does not admit multimorphism
$\mmorp{\min, \min}$ and denote its arity by $r$. Let $\tup{x}, \tup{y} \in
\Feas(\gamma)$ be $r$-tuples such that $\gamma(\tup{x}) + \gamma(\tup{y}) < 2
\cdot \gamma(\min(\tup{x}, \tup{y}))$. Without loss of generality, we have
$\gamma(\tup{x}) < \gamma(\min(\tup{x}, \tup{y}))$ and may assume that $\tup{y}
= \min(\tup{x}, \tup{y})$. Again, $\tup{x}$ and $\tup{y}$ must differ at every
coordinate, which implies $\tup{x} = \tup{1}^r, \tup{y} = \tup{0}^r$. If $r \geq
2$, we would obtain a contradiction by applying the $=$-restriction and
minimisation at coordinate $1$. Hence, $r = 1$ and by scaling and adding a
constant to $\gamma$ we get $\gamma_1 \in \wClonep{\Gamma}$.

Analogously, if $\max \not\in \Pol(\wClonep{\Gamma})$, we get $\rho_{\uparrow}'
\in \wClonep{\Gamma}$ where $\rho_{\uparrow}'$ is a binary relation such that
$\rho_{\neq} \subseteq \rho_{\uparrow}' \subseteq \rho_{\neq} \cup \{ (0, 0)
\}$. Otherwise, $\gamma_0 \in \wClonep{\Gamma}$. It holds
\begin{align}
\rho_{\neq}(x, y)
&= \rho_{\vee}'(x, y) + \rho_{\uparrow}'(x, y) \\
&= \Opt \left( \rho_{\vee}'(x, y) + \gamma_0(x) + \gamma_0(y) \right) \\
&= \Opt \left( \rho_{\uparrow}'(x, y) + \gamma_1(x) + \gamma_1(y) \right) \,,
\end{align}
so $\rho_{\neq}$ can be constructed with a planar gadget if at least one of
$\min$, $\max$ is not a polymorphism of $\wClonep{\Gamma}$.

Finally, consider the case when $\min, \max \in \Pol(\wClonep{\Gamma})$ and
hence $\gamma_0, \gamma_1 \in \wClonep{\Gamma}$. Set
$\wClonep{\Gamma}$ is then a conservative language, so we have
$\closure{\wClonep{\Gamma}} = \wClonep{\Gamma}$. We choose a minimum-arity
weighted relation $\gamma \in \wClonep{\Gamma}$ that does not admit
multimorphism $\mmorp{\min, \max}$ and denote its arity by $r$. Let $\tup{x},
\tup{y} \in \Feas(\gamma)$ be tuples such that $\gamma(\tup{x}) +
\gamma(\tup{y}) < \gamma(\min(\tup{x}, \tup{y})) + \gamma(\max(\tup{x},
\tup{y}))$. Note that $\min(\tup{x}, \tup{y}), \max(\tup{x}, \tup{y}) \in
\Feas(\gamma)$. By the choice of $\gamma$, tuples $\tup{x}, \tup{y}$ must
differ at every coordinate, and hence
$\tup{y} = \negate{\tup{x}}$, $\min(\tup{x}, \tup{y}) = \tup{0}^r$,
$\max(\tup{x}, \tup{y}) = \tup{1}^r$. We partition coordinates $\{1, \dots,
r\}$ into $I = \{ i ~|~ x_i = 0 \}$ and $J = \{ j ~|~ x_j = 1 \}$. By
Lemma~\ref{lmSwapping}, $\closure{\{\gamma\}} \subseteq \wClonep{\Gamma}$ contains a \emph{binary} weighted relation that
does not admit multimorphism $\mmorp{\min, \max}$, and hence $r = 2$. It holds
$\gamma(0, 1) + \gamma(1, 0) < \gamma(0, 0) + \gamma(1, 1)$, where all the
values are finite. We may assume that $\gamma(0, 0) + \gamma(1, 1) - \gamma(0,
1) - \gamma(1, 0) = 2$ and $\gamma(0, 0) = 1$ (this can be achieved by scaling
and adding a constant). We define unary weighted relations $\mu_1, \mu_2 \in
\wClonep{\Gamma}$ as $\mu_1(0) = \mu_2(0) = 0$, $\mu_1(1) = -\gamma(1, 0)$,
$\mu_2(1) = -\gamma(0, 1)$. By adding $\mu_1$ and $\mu_2$ to $\gamma$ at the
first and second coordinate respectively we get $\gamma_{\neq}$, and therefore
$\rho_{\neq} = \Opt(\gamma_{\neq}) \in \wClonep{\Gamma}$.
\end{proof}

\begin{lemma}
\label{lmNegation}
Let $\Gamma$ be a valued constraint language that does not admit multimorphism $\mmorp{\neg}$. If
$\rho_{\neq} \in \wClonep{\Gamma}$, then $\rho_0, \rho_1 \in \wClonep{\Gamma}$.
\end{lemma}

\begin{proof}
We choose a minimum-arity weighted relation $\gamma \in \wClonep{\Gamma}$ that
does not admit multimorphism $\mmorp{\neg}$ and denote its arity by $r$. Let
$\tup{x} \in \Feas(\gamma)$ be an $r$-tuple such that $\gamma(\tup{x}) \neq
\gamma(\negate{\tup{x}})$. It must hold $r = 1$, otherwise we would get a
smaller counterexample by applying the $=$-restriction or $\neq$-restriction at the first
coordinate (depending on whether $x_1 = x_2$ or $x_1 \neq x_2$) followed by
minimisation. Hence, $\Opt(\gamma) = \rho_0$ or $\Opt(\gamma) = \rho_1$.
Say $\Opt(\gamma)=\rho_0$, the other case is analogous. Then the twist
$\gamma'(x)=\gamma(x\oplus 1)$ of $\gamma$ satisfies $\Opt(\gamma')=\rho_1$.
\end{proof}

\begin{lemma}
\label{lmMnrtMjrt}
Let $\Gamma$ be a valued constraint language that admits neither of the multimorphisms
$\mmorp{\mnrt, \mnrt, \mnrt}$, $\mmorp{\mjrt, \mjrt, \mjrt}$, $\mmorp{\mjrt,
\mjrt, \mnrt}$. If $\rho_0, \rho_1, \rho_{\neq} \in \wClonep{\Gamma}$, then
$\rho_\text{1-in-3} \in \wClonep{\Gamma}$.
\end{lemma}

\begin{proof}
If $\mnrt \not\in \Pol(\wClonep{\Gamma})$, we choose a minimum-arity
relation $\rho \in \wClonep{\Gamma}$ that is not invariant under $\mnrt$. Its
arity $r$ must be at least $2$; let us first assume $r \geq 3$. For any triple
of $r$-tuples from $\rho$ that agree at some coordinate, the $r$-tuple obtained
by applying $\mnrt$ to them also belongs to $\rho$ (otherwise we would get a
contradiction with the choice of $\rho$ by taking a pinning instead). Let $\tup{x}, \tup{y}, \tup{z} \in \rho$ be
$r$-tuples such that $\mnrt(\tup{x}, \tup{y}, \tup{z}) \not\in \rho$. Without
loss of generality, we may assume that $\mnrt(\tup{x}, \tup{y}, \tup{z}) =
\tup{0}^r$ (this can be achieved with twists). By the same argument as in
Lemma~\ref{lmConstants}, we have $\base{i}{r} \in \rho$ for all $i$. Let
$\tup{w} \in \rho$ be a tuple with the minimum even number of ones (such a
tuple exists as at least one of $\tup{x}, \tup{y}, \tup{z}$ contains an even
number of ones). If $\tup{w} \neq \tup{1}^r$, there are distinct coordinates
$i, j, k$ with $w_i = w_j = 1, w_k = 0$. Because $\tup{w}, \base{i}{r},
\base{j}{r}$ agree at coordinate $k$, tuple $\mnrt(\tup{w}, \base{i}{r},
\base{j}{r})$ belongs to $\rho$. However, it has two fewer ones than $\tup{w}$,
which is a contradiction. Hence, $\tup{w} = \tup{1}^r$ and $r \geq 4$.
But then $\mnrt(\tup{1}^r, \base{3}{r}, \base{4}{r}) \not\in \rho$ (as it
contains an even number of ones), and we obtain a smaller counterexample by
taking the $=$-restriction of $\rho$ at the first coordinate followed by
minimisation. Therefore, $r = 2$ and $|\rho| = 3$. Using
twists, we can get from $\rho$ relation $\rho_{\uparrow} = \{ (0, 0), (0, 1),
(1, 0) \} \in \wClonep{\Gamma}$.

If $\mjrt \not\in \Pol(\wClonep{\Gamma})$, we choose a minimum-arity relation
$\rho_\text{1-in-3}' \in \wClonep{\Gamma}$ that is not invariant under $\mjrt$.
Its arity $r$ must be at least $3$ since every unary and binary relation admits
$\mjrt$ as a polymorphism. By the same argument as for $\mnrt$, we may assume
$\tup{0}^r \not\in \rho_\text{1-in-3}'$, and it can be shown that $\base{i}{r}
\in \rho_\text{1-in-3}'$ for all $i$. If $r \geq 4$, tuples $\base{1}{r},
\base{2}{r}, \base{3}{r}$ and $\mjrt(\base{1}{r}, \base{2}{r}, \base{3}{r}) =
\tup{0}^r$ agree at coordinate $4$; we then obtain a smaller counterexample by
pinning. Therefore, $r = 3$.

If neither of $\mnrt, \mjrt$ is a polymorphism of $\wClonep{\Gamma}$, we have
\begin{equation}
\rho_\text{1-in-3}(x, y, z)
= \rho_\text{1-in-3}'(x, y, z) +
  \rho_{\uparrow}(x, y) + \rho_{\uparrow}(y, z) + \rho_{\uparrow}(z, x) \,,
\end{equation}
which can be implemented in a planar way, and hence $\rho_\text{1-in-3} \in
\wClonep{\Gamma}$. Otherwise, $\Gamma$ is not a crisp language (i.e.\ not a
$\{0,\infty\}$-valued language), because that would make it admit multimorphism
$\mmorp{\mnrt, \mnrt, \mnrt}$ or $\mmorp{\mjrt, \mjrt, \mjrt}$. Let $\mu \in
\wClonep{\Gamma}$ be a minimum-arity non-crisp weighted relation and $\tup{x},
\tup{y} \in \Feas(\mu)$ tuples such that $\mu(\tup{x}) \neq \mu(\tup{y})$.
Tuples $\tup{x}, \tup{y}$ differ at every coordinate (otherwise we could obtain
a smaller counterexample by pinning), and hence $\tup{y} = \negate{\tup{x}}$.
Moreover, $\mu$ is unary, otherwise we could apply the $=$-restriction or
$\neq$-restriction at the first coordinate (depending on whether $x_1 = x_2$ or
$x_1 \neq x_2$) followed by minimisation to obtain a smaller counterexample. If
$\mu(0) < \mu(1)$, we get $\gamma_0 \in \wClonep{\Gamma}$ by scaling $\mu$ and
adding a constant, and $\gamma_1 \in \wClonep{\Gamma}$ by twisting $\gamma_0$;
the case $\mu(0) > \mu(1)$ is symmetric.
It holds
\begin{align}
\rho_\text{1-in-3}(x, y, z)
&= \Opt \left(
      \rho_{\uparrow}(x, y) + \rho_{\uparrow}(y, z) + \rho_{\uparrow}(z, x) +
      \gamma_1(x) + \gamma_1(y) + \gamma_1(z) \right) \\
&= \Opt \left( \rho_\text{1-in-3}'(x, y, z) +
      \gamma_0(x) + \gamma_0(y) + \gamma_0(z) \right) \,.
\end{align}
Both can be implemented planarly, and therefore $\rho_\text{1-in-3} \in
\wClonep{\Gamma}$ if exactly one of $\mnrt, \mjrt$ is a polymorphism of
$\wClonep{\Gamma}$.

Finally, we consider the case when both $\mnrt, \mjrt \in \Pol(\wClonep{\Gamma})$.
Let $\gamma \in \wClonep{\Gamma}$ be an $r$-ary weighted relation of the minimum
arity for which Inequality~\eqref{eqMMorpIneq} does not hold as \emph{equality}
for multimorphism $\mmorp{\mjrt, \mjrt, \mnrt}$. Let $\tup{x}, \tup{y}, \tup{z}
\in \Feas(\gamma)$ be $r$-tuples that violate the equality. They do not agree at
any coordinate (otherwise we could obtain a smaller counterexample by pinning),
and hence $\mjrt(\tup{x}, \tup{y}, \tup{z})$ and $\mnrt(\tup{x}, \tup{y},
\tup{z})$ differ everywhere. Without loss of generality, we may assume that
$\mjrt(\tup{x}, \tup{y}, \tup{z}) = \tup{0}^r$ and $\mnrt(\tup{x}, \tup{y},
\tup{z}) = \tup{1}^r$ (this can be achieved with twists) and $\tup{z} \neq
\tup{0}^r$. Note that $\tup{0}^r, \tup{1}^r \in \Feas(\gamma)$ because $\mjrt,
\mnrt$ are polymorphisms of $\gamma$. Tuples $\tup{x}, \tup{y}, \tup{0}^r$ agree
at all coordinates $i$ where $z_i = 1$, and hence they satisfy
\eqref{eqMMorpIneq} as equality, i.e.
\begin{equation}
\gamma(\tup{x}) + \gamma(\tup{y}) + \gamma(\tup{0}^r)
= 2\cdot\gamma(\mjrt(\tup{x}, \tup{y}, \tup{0}^r))
  + \gamma(\mnrt(\tup{x}, \tup{y}, \tup{0}^r))
= 2\cdot\gamma(\tup{0}^r) + \gamma(\negate{\tup{z}}) \,.
\end{equation}
Because $\gamma(\tup{x}) + \gamma(\tup{y}) + \gamma(\tup{z}) \neq
2\cdot\gamma(\tup{0}^r) + \gamma(\tup{1}^r)$, this implies $\gamma(\tup{z}) +
\gamma(\negate{\tup{z}}) \neq \gamma(\tup{0}^r) + \gamma(\tup{1}^r)$. We are
going to apply Lemma~\ref{lmSwapping} for this disequality. Language
$\wClonep{\Gamma}$ is conservative (as it contains both $\gamma_0, \gamma_1$),
and hence $\closure{\wClonep{\Gamma}} = \wClonep{\Gamma}$. It admits a majority
polymorphism, therefore every relation in $\wClonep{\Gamma}$ is
2-decomposable~\cite{Jeavons98:consist}. We partition coordinates $\{1, \dots,
r\}$ into $I = \{ i ~|~ z_i = 0\}$ and $J = \{ j ~|~ z_j = 1\}$. By
Lemma~\ref{lmSwapping}, there is a binary weighted relation $\gamma' \in
\closure{\{\gamma\}} \subseteq \wClonep{\Gamma}$ with $\Feas(\gamma') = D^2$ and
$\gamma'(0,1) + \gamma'(1,0) \neq \gamma'(0,0) + \gamma'(1,1)$. We may assume
that $\gamma'(0,1) + \gamma'(1,0) < \gamma'(0,0) + \gamma'(1,1)$, otherwise we
apply a twist. As in the proof of Lemma~\ref{lmMinMax}, weighted relation
$\gamma_{\neq}$ can be obtained from $\gamma'$. Then we planarly construct
$\rho_\text{1-in-3} \in \wClonep{\Gamma}$ as
\begin{equation}
\rho_\text{1-in-3}(x, y, z)
= \Opt \left( \gamma_{\neq}(x, y) + \gamma_{\neq}(y, z) + \gamma_{\neq}(z, x) +
              \gamma_0(x) + \gamma_0(y) + \gamma_0(z) \right) \,.
\end{equation}
\end{proof}

\begin{proof}[Proof (of Theorem~\ref{thm:main1})]
Since $\Gamma$ is intractable we know,
by~\cite[Theorem~7.1]{Cohen06:complexitysoft}, that $\Gamma$ admits neither of
the multimorphisms $\mmorp{c_0}$, $\mmorp{c_1}$, $\mmorp{\min, \min}$,
$\mmorp{\max, \max}$, $\mmorp{\min, \max}$, $\mmorp{\mnrt, \mnrt, \mnrt}$,
$\mmorp{\mjrt, \mjrt, \mjrt}$, $\mmorp{\mjrt, \mjrt, \mnrt}$. By assumption,
$\Gamma$ is not self-complementary and hence does not admit the $\mmorp{\neg}$
multimorphism.

By Lemmas~\ref{lmConstants},~\ref{lmMinMax}, and~\ref{lmNegation}, we have $\rho_0, \rho_1, \rho_{\neq}
\in \wClonep{\Gamma}$ and hence by Lemma~\ref{lmMnrtMjrt} $\rho_\text{1-in-3} \in
\wClonep{\Gamma}$. Planar \class{$1$-in-$3$ Positive $3$-Sat}
problem is NP-complete~\cite{Mulzer08:jacm}, and therefore
$\Gamma$ is planarly-intractable by Theorem~\ref{thmpClone}.
\end{proof}

\section{Conservative Valued CSPs}
\label{sec:cons}

A valued constraint language $\Gamma$ is called \emph{conservative} if $\Gamma$
includes all $\{0,1\}$-valued unary weighted relations. As our second result, we
prove that planarity does not add any tractable cases for conservative valued
constraint languages. 

\begin{theorem}\label{thm:main2}
Let $\Gamma$ be an intractable conservative valued constraint language. Then
$\Gamma$ is planarly-intractable.
\end{theorem}

Consequently, we obtain a complexity classification of all conservative valued
constraint languages in the planar setting, thus sharpening the classification
of conservative valued constraint languages~\cite{kz13:jacm,tz15:icalp}. As
mentioned in Section~\ref{sec:intro}, for Boolean domains
Theorem~\ref{thm:main2} follows from Theorem~\ref{thm:main1}. Thus, the only
tractable Boolean conservative languages in the planar setting are given by
the multimorphisms $\mmorp{\min,\max}$ and
$\mmorp{\mjrt,\mjrt,\mnrt}$~\cite{Cohen06:complexitysoft}.

We now define certain special types of multimorphisms.

A $k$-ary operation $f:D^k\to D$ if called \emph{conservative} if
$f(x_1,\ldots,x_k)\in\{x_1,\ldots,x_k\}$ for every $x_1,\ldots,x_k\in D$.
A multimorphism $\mmorp{f_1,\ldots,f_k}$ is called \emph{conservative} if
applying $\mmorp{f_1,\ldots,f_k}$ to $\tuple{x_1,\ldots,x_k}$ returns a
permutation of $\tuple{x_1,\ldots,x_k}$.

\begin{definition}
A binary multimorphism $\mmorp{f,g}$ of $\Gamma$ is called a
\emph{symmetric tournament pair} (STP) if it is conservative and both $f$ and
$g$ are commutative operations.
\end{definition}

It was shown in~\cite{Cohen08:Generalising} that languages admitting an STP
multimorphism are tractable.

\begin{definition}
A ternary multimorphism $\mmorp{f,g,h}$ is called an MJN if $f$ and $g$ are
(possibly equal) majority operations and $h$ is a minority operation.
\end{definition}

It was shown in~\cite{kz13:jacm} that languages admitting an MJN multimorphism
are tractable.

\begin{theorem}[\cite{kz13:jacm}]\label{thm:cons}
Let $\Gamma$ be a conservative valued constraint language on $D$.
Then either $\Gamma$ admits a conservative binary multimorphism $\mmorp{f,g}$ and a conservative
ternary multimorphism $\mmorp{f',g',h'}$ and there is a family $M$ of 2-element subsets of $D$, such that

\begin{itemize}
\item
for every $\{a,b\}\in M$, $\mmorp{f,g}$ restricted to $\{a,b\}$ is a symmetric
tournament pair, and
\item for every $\{a,b\}\not \in M$, $\mmorp{f',g',h'}$ restricted to $\{a,b\}$
is an MJN multimorphism,
\end{itemize}

in which case $\Gamma$ is tractable,
or else $\Gamma$ is intractable.
\end{theorem}

The idea of the proof of Theorem~\ref{thm:cons} is as follows: given a conservative
valued constraint language $\Gamma$, we define a certain graph $G_\Gamma$ whose
vertices are pairs of different labels from $D$ and an edge $(a,b) - (c,d)$ is present if
there is a binary weighted relation $\gamma\in\wCloneg{\Gamma}$ that is
``non-submodular with respect to the order $a<b$ and $c<d$''. The edges of
$G_\Gamma$ are then classified as soft and hard. It is shown that a soft
self-loop implies intractability of $\Gamma$. Otherwise, the vertices of
$G_\Gamma$ are partitioned into $M\cup\overline{M}$, where $M$ denotes the set
of loopless vertices and $\overline{M}$ denotes the rest (i.e.\ vertices with
hard loops). It is then shown that $G_\Gamma$ restricted to $M$ is bipartite,
which is in turn used to construct a binary multimorphism and a ternary
multimorphism of $\Gamma$ such that the binary multimorphism is an STP on $M$
and the ternary multimorphism is an MJN on $\overline{M}$. (Proving that the
constructed objects are multimorphisms of $\Gamma$ is the most
technical part of the proof.) Any such language is
then tractable via an involved algorithm from~\cite{kz13:jacm} that relies
on~\cite{Cohen08:Generalising}, or by an LP relaxation~\cite{tz15:icalp}.

Our approach is to follow the above-described proof and adapt it to the planar
setting. 
We remark that similar graphs to $G_\Gamma$ have been important in other
studies of (V)CSPs. In particular, in the classification of conservative
CSPs~\cite{Bulatov11:conservative} and in the classification of Minimum Cost
Homomorphism problems~\cite{Takhanov10:stacs}. In~\cite{Bulatov11:conservative},
the graph has labels as vertices and three types of edges depending on three
types of polymorphisms. In~\cite{Takhanov10:stacs}, the graph has, as in our
case, pairs of labels as vertices but the edges of the graph are defined,
informally, via a $\min$/$\max$ polymorphism rather than a $\mmorp{\min,\max}$
multimorphism. Also, edges in~\cite{Takhanov10:stacs} are not classified as
soft or hard.

It is natural to replace $\wCloneg{\Gamma}$ by $\wClonep{\Gamma}$ in the
definition of $G_\Gamma$. But this simple change does not immediately yield the
desired result. There are two main obstacles. First, the proof of
Theorem~\ref{thm:cons} from~\cite{kz13:jacm} heavily relies
on~\cite{Takhanov10:stacs}, which guarantees, unless in an NP-hard case, the
existence of a majority polymorphism and hence that the language is 2-decomposable. Second, some of the gadgets (and in
particular the ``$i$-expansion'' from~\cite[Section~6.4]{kz13:jacm}) are not
necessarily planar. In more detail, \cite{Takhanov10:stacs} builds a similar
graph to ours (as described above) and argues that, unless in
an NP-hard case, this graph is bipartite (part of our $G_\Gamma$ will also be
bipartite). This property is then used
in~\cite{Takhanov10:stacs} to argue about the existence of a majority
polymorphism. However, this is proved in~\cite{Takhanov10:stacs} using clones
and depends on the Galois connection between clones and relational co-clones;
such a connection is not known for planar expressibility!

To avoid these obstacles, we modify, significantly simplify, and generalise the
proof so that it works in the planar setting. The key changes are the following.
(i) We define our graph based on a language closure $\closure{\Gamma}$, which is
a subset of the planar weighted relational clone of a conservative language. (ii) We do \emph{not} rely
on Takhanov's result on the existence of a majority polymorphism~\cite{Takhanov10:stacs} but
instead prove directly that (the closure
of) $\Gamma$ is 2-decomposable. (iii) We define different STP and MJN
multimorphisms that allow us to simplify the proof that these are indeed
multimorphisms of $\Gamma$. In particular, we will prove modularity
of weighted relations on $\overline{M}$ and
show that the ternary MJN multimorphism satisfies Inequality~\eqref{eqMMorpIneq}
with equality, thus obtaining a better structural understanding of tractable
languages. The main simplification is that we define MJN as close to projection
operations as possible, and in particular not depending on the STP multimorphism as
in~\cite{kz13:jacm}.

We remark that it is not clear how to derive non-trivial properties of graph
$G_\Gamma$ used in our proofs from the related graph defined in~\cite{kz13:jacm}
apart from the obvious fact that our graph is a subgraph of the graph
from~\cite{kz13:jacm}. We believe that with more work one can derive that the
two graphs are in fact the same using techniques and proofs from this paper, but have
not done so since our goal was to obtain a complexity classification.

The rest of this section is devoted to proving Theorem~\ref{thm:main2}.

\begin{definition}
Let $\Gamma$ be a conservative language. We define an undirected graph
$G_\Gamma$ on vertices $(a, b)$ for all $a, b \in D, a \neq b$. For any vertex
$v = (a, b)$, we will denote by $\overline{v}$ vertex $(b, a)$. Graph $G_\Gamma$
is allowed to have self-loops. It contains edge $(a_1, b_1) - (a_2, b_2)$ if
there is a binary weighted relation $\gamma \in \clGamma$ such that $(a_1, b_2),
(b_1, a_2) \in \Feas(\gamma)$ and
\begin{equation}
\label{eqEdgeDefinition}
\gamma(a_1, b_2) + \gamma(b_1, a_2) < \gamma(a_1, a_2) + \gamma(b_1, b_2) \,.
\end{equation}
If there exists such a weighted relation $\gamma$ with at least one of $(a_1,
a_2), (b_1, b_2)$ belonging to $\Feas(\gamma)$, we will call the edge
\emph{soft}, otherwise the edge is \emph{hard}. We denote by $\overline{M}$ and
$M$ the set of vertices with and without self-loops respectively.
\end{definition}

The following lemma gives a useful alternative characterisation of an edge in $G_\Gamma$.

\begin{lemma}
\label{lmEdgeRedefined}
Graph $G_\Gamma$ contains edge $(a_1, b_1) - (a_2, b_2)$ if, and only if, binary
relation $\{ (a_1, b_2), (b_1, a_2) \}$ belongs to $\clGamma$. The edge is soft
if, and only if, at least one of binary relations $\{ (a_1, a_2), (a_1, b_2),
(b_1, a_2) \}$, $\{ (b_1, b_2), (a_1, b_2), (b_1, a_2) \}$ belongs to
$\clGamma$.
\end{lemma}

\begin{proof}
Both \emph{if} implications follow directly from the definition of $G_\Gamma$;
we need to prove the \emph{only if} part. Let $\gamma$ be a weighted relation
establishing edge $(a_1, b_1) - (a_2, b_2)$ such that $\Feas(\gamma) \subseteq
\{a_1, b_1\} \times \{a_2, b_2\}$ (this can be always achieved by domain
restriction). Note that we may add to $\gamma$ any unary finite-valued weighted
relation without invalidating \eqref{eqEdgeDefinition}. We choose any
$\lambda \in \Q$ such that $\lambda < \gamma(b_1, b_2)$ and $\gamma(a_1, b_2) +
\gamma(b_1, a_2) - \lambda < \gamma(a_1, a_2)$. Note that such $\lambda$ exists
due to \eqref{eqEdgeDefinition}. We define unary weighted relations $\gamma_1,
\gamma_2$ such that $\gamma_1(a_1) = \lambda - \gamma(a_1, b_2)$, $\gamma_2(a_2)
= \lambda - \gamma(b_1, a_2)$, and $\gamma_1(x) = \gamma_2(x) = 0$ otherwise.
Now consider binary weighted relation $\gamma'$ defined as $\gamma'(x, y) =
\gamma(x, y) + \gamma_1(x) + \gamma_2(y)$. We have $\gamma'(a_1, b_2) =
\gamma'(b_1, a_2) = \lambda$ and $\lambda < \gamma'(a_1, a_2), \gamma'(b_1,
b_2)$, so then $\Opt(\gamma') = \{ (a_1, b_2), (b_1, a_2) \} \in \clGamma$.

If the edge is soft and $(a_1, a_2), (b_1, b_2) \in \Feas(\gamma)$, we proceed
as above with $\lambda = \gamma(b_1, b_2)$, so that $\Opt(\gamma') = \{ (b_1,
b_2), (a_1, b_2), (b_1, a_2) \} \in \clGamma$. Otherwise we simply take
$\Feas(\gamma) \in \clGamma$.
\end{proof}

We show that the absence of soft self-loops is a necessary condition for
planar tractability.

\begin{theorem}\label{thm:ssl}
If $G_\Gamma$ has a soft self-loop, language $\Gamma$ is planarly-intractable.
\end{theorem}

\begin{proof}
Let $(a, b)$ be a vertex of $G_\Gamma$ with a soft self-loop. Without loss of
generality, we have $\rho = \{ (a, a), (a, b), (b, a) \} \in \clGamma$ by
Lemma~\ref{lmEdgeRedefined}. We denote by $\gamma_a, \gamma_b$ the unary weighted
relations defined as $\gamma_a(a) = \gamma_b(b) = 0$, $\gamma_a(b) = \gamma_b(a)
= 1$, and $\gamma_a(x) = \gamma_b(x) = \infty$ for $x \not\in \{a, b\}$. Set
$\Gamma' = \{\rho, \gamma_a, \gamma_b\} \subseteq \clGamma$ can be viewed as a
conservative language over a Boolean domain $\{a, b\}$. 
By~\cite[Theorem~7.1]{Cohen06:complexitysoft}, $\Gamma'$ is intractable (in
particular, $\Gamma'$ does not fall into either of the two
tractable cases for Boolean conservative valued constraint
languages~\cite{Cohen06:complexitysoft} corresponding to the $\mmorp{\min,\max}$
and $\mmorp{\mjrt,\mjrt,\mnrt}$ multimorphisms).
Observe that $\Gamma'$ is not self-complementary
since neither of its weighted relations is self-complementary.
By Theorem~\ref{thm:main1}, $\Gamma'$ is planarly-intractable and thus, by
Theorem~\ref{thmpClone}, so is $\Gamma$.
\end{proof}

It remains to show that this condition is also sufficient.

\begin{theorem}\label{thm:goal}
If $G_\Gamma$ has no soft self-loop, then $\Gamma$ admits a binary
multimorphism $\STP$ that is an STP on $M$, and a ternary multimorphism $\MJN$
that is an MJN on $\overline{M}$.
\end{theorem}

In order to prove Theorem~\ref{thm:goal}, we now introduce several lemmas. From now on we
will assume that $G_\Gamma$ has no soft self-loop.

\begin{lemma}
\label{lmGraphProperties}
For any vertex $v$, graph $G_\Gamma$ contains edge $v - \overline{v}$. There is
no edge between $M$ and $\overline{M}$, no odd cycle in $M$, and no soft edge in
$\overline{M}$.
\end{lemma}

\begin{proof}
As the binary equality relation belongs to $\clGamma$, we have edge $v -
\overline{v}$ for all vertices $v$.

Consider any sequence of vertices $v_1, v_2, v_3, v_4$ such that there is
an edge between every two consecutive ones, and denote $v_i = (a_i, b_i)$. By
Lemma~\ref{lmEdgeRedefined}, there exist binary relations $\rho_i = \{ (a_i,
b_{i+1}), (b_i, a_{i+1}) \} \in \clGamma$ for $i \in \{1, 2, 3\}$. Their join
equals $\{ (a_1, b_4), (b_1, a_4) \} \in \clGamma$, and hence $G_\Gamma$
contains edge $v_1 - v_4$. If any of edges $v_1 - v_2, v_2 - v_3, v_3 - v_4$ is
soft, we can replace the corresponding relation $\rho_i$ with $\{ (a_i,
a_{i+1}), (a_i, b_{i+1}), (b_i, a_{i+1}) \}$ or $\{ (b_i, b_{i+1}), (a_i,
b_{i+1}), (b_i, a_{i+1}) \}$ to show that $v_1 - v_4$ is also soft.

Suppose that there is an edge between $s \in M$ and $t \in \overline{M}$. Then
we have edges $s - t, t - t, t - s$, and hence also self-loop $s - s$, which is
a contradiction.

If there is an odd cycle in $M$, let us choose a shortest one and denote its
vertices $v_1, \dots, v_k$ ($k \geq 3$). We have a sequence of adjacent vertices
$v_k, v_1, v_2, v_3$, and hence $v_3$ and $v_k$ are also adjacent. But that
means there is a shorter odd cycle (or a self-loop) $v_3, \dots, v_k$; a
contradiction.

Finally, suppose that $s, t \in \overline{M}$ and edge $s - t$ is soft. Then we
have edges $s - t, t - t, t - s$, and hence a soft self-loop at $s$, which is a
contradiction.
\end{proof}

\begin{lemma}
\label{lmGamma2Decomposable}
Every relation in $\clGamma$ is 2-decomposable.
\end{lemma}

\begin{proof}
Let $\rho \in \clGamma$ be an $r$-ary relation. By definition, $\tup{x} \in \rho$
implies $\bigwedge_{1\leq i,j\leq r} (x_i, x_j) \in \Projection{i,j}{\rho}$. We
prove the converse implication by induction on $r$. If $r \leq 2$, relation
$\rho$ is trivially 2-decomposable. Let $r = 3$. Suppose for the
sake of contradiction that $(x_1, x_2, x_3) \not\in \rho$ even though $(y_1,
x_2, x_3), (x_1, y_2, x_3), (x_1, x_2, y_3) \in \rho$ for some $y_1, y_2, y_3
\in D$. Let $\rho_1 \in \clGamma$ be the binary relation obtained from $\rho$ by
pinning it at the first coordinate to label $x_1$; we have $(x_2, y_3), (y_2,
x_3) \in \rho_1$, $(x_2, x_3) \not\in \rho_1$, and thus graph $G_\Gamma$
contains edge $(x_2, y_2) - (x_3, y_3)$. Analogously, the graph contains edges
$(x_3, y_3) - (x_1, y_1)$ and $(x_1, y_1) - (x_2, y_2)$. This is an odd cycle,
so it must hold $(x_1, y_1), (x_2, y_2), (x_3, y_3) \in \overline{M}$. Let
$\gamma$ be a unary weighted relation with $\gamma(x_1) = 0, \gamma(y_1) = 1$
and $\gamma(z) = \infty$ for all $z \in D \setminus \{ x_1, y_1 \}$. By adding
$\gamma$ to $\rho$ at the first coordinate and then minimising over it we show
that edge $(x_2, y_2) - (x_3, y_3)$ is soft, which is a contradiction.

It remains to prove the lemma for $r \geq 4$. Let $\rho_1 \in \clGamma$ be the
relation obtained from $\rho$ by minimisation over the first coordinate. Relation $\rho_1$ is
2-decomposable by the induction hypothesis, so $(x_2, \dots, x_r) \in \rho_1$,
and hence $(y_1, x_2, \dots, x_r) \in \rho$ for some $y_1 \in D$. Analogously,
we have $(x_1, y_2, x_3, \dots, x_r), (x_1, x_2, y_3, x_4, \dots, x_r) \in \rho$
for some $y_2, y_3 \in D$. Pinning $\rho$ at every coordinate $k \geq 4$ to its
respective label $x_k$ gives a ternary 2-decomposable relation $\rho'$ such
that $(x_i, x_j) \in \Projection{i,j}{\rho'}$ for all $i, j \in \{1, 2, 3\}$.
Therefore, $(x_1, x_2, x_3) \in \rho'$ and $\tup{x} \in \rho$.
\end{proof}

The following lemma involves an extension of the definition of an edge in
$G_\Gamma$ to non-binary weighted relations.

\begin{lemma}
\label{lmEdgeRAry}
Let $\gamma \in \clGamma$ be an $r$-ary weighted relation and $I, J \subseteq
\{1, \dots, r\}$ a partition of its coordinates. If $\tup{x}, \tup{y} \in
\Feas(\gamma)$ and
\begin{equation}
\gamma(\tup{x}) + \gamma(\tup{y}) <
\gamma(\joinIJ{x}{y}) + \gamma(\joinIJ{y}{x}) \,,
\end{equation}
then graph $G_\Gamma$ contains edge $(x_i, y_i) - (y_j, x_j)$ for some $i \in I,
j \in J$. If at least one of $\joinIJ{x}{y}, \joinIJ{y}{x}$ belongs to
$\Feas(\gamma)$, the edge is soft.
\end{lemma}

\begin{proof}
By Lemma~\ref{lmSwapping}, there are coordinates $i \in I, j \in J$ and a binary
weighted relation $\gamma_{i,j} \in \clGamma$ such that $(x_i, x_j), (y_i, y_j)
\in \Feas(\gamma_{i,j})$ and $\gamma_{i,j}(x_i, x_j) + \gamma_{i,j}(y_i, y_j) <
\gamma_{i,j}(x_i, y_j) + \gamma_{i,j}(y_i, x_j)$, so graph $G_\Gamma$ contains
edge $(x_i, y_i) - (y_j, x_j)$.  If $\joinIJ{x}{y}$ or $\joinIJ{y}{x}$ belongs
to $\Feas(\gamma)$, then $(x_i, y_j)$ or $(y_i, x_j)$ belongs to
$\Feas(\gamma_{i,j})$ (as $\closure{\{\gamma\}}$ is 2-decomposable by
Lemma~\ref{lmGamma2Decomposable}), and hence the edge is soft.
\end{proof}

\begin{lemma}
\label{lmMBarSwappingEq}
Let $\gamma \in \clGamma$ be an $r$-ary weighted relation and $I, J \subseteq
\{1, \dots, r\}$ a partition of its coordinates. If $\tup{x}, \tup{y},
\joinIJ{x}{y}, \joinIJ{y}{x} \in \Feas(\gamma)$ and, for all $i \in I$, $(x_i,
y_i) \in \overline{M}$, then
\begin{equation}
\gamma(\tup{x}) + \gamma(\tup{y}) =
\gamma(\joinIJ{x}{y}) + \gamma(\joinIJ{y}{x}) \,.
\end{equation}
\end{lemma}

\begin{proof}
Suppose for the sake of contradiction that the equality does not hold. Without
loss of generality, we may assume that $\gamma(\tup{x}) + \gamma(\tup{y}) <
\gamma(\joinIJ{x}{y}) + \gamma(\joinIJ{y}{x})$. By Lemma~\ref{lmEdgeRAry}, graph
$G_\Gamma$ contains a soft edge incident to $(x_i, y_i)$ for some $i \in I$,
which contradicts Lemma~\ref{lmGraphProperties}.
\end{proof}

Graph $G_\Gamma$ does not have any odd cycle on vertices $M$. Therefore, there
is a partition of $M$ into two independent sets $M_1, M_2$. (In fact, it
can be shown that every connected component of $G_\Gamma$ restricted to $M$
is a complete bipartite graph but we do not need this fact here.)
Note that $(a, b)
\in M_1$ if, and only if, $(b, a) \in M_2$, as every vertex $v \in M$ is adjacent
to $\overline{v}$. We define multimorphism $\STP$ as follows:
\begin{subnumcases}{\STP(x, y) =}
\label{eqSTP_M1}
(x, y) & \text{if $(x, y) \in M_1$,} \\
\label{eqSTP_M2}
(y, x) & \text{if $(x, y) \in M_2$,} \\
\label{eqSTP_otherwise}
(x, y) & \text{otherwise.}
\end{subnumcases}
By definition, $\STP$ is commutative on $M$.

\begin{theorem}
\label{thmSTP}
$\STP$ is a multimorphism of $\Gamma$.
\end{theorem}

\begin{proof}
Let $\gamma \in \Gamma$ be an $r$-ary weighted relation and $\tup{x}, \tup{y}
\in \Feas(\gamma)$. Suppose for the sake of contradiction that
\eqref{eqMMorpIneq} does not hold. We partition set $\{1, \dots, r\}$ into
$I$ and $J$: Set $J$ consists of all coordinates $j$ such that case
\eqref{eqSTP_M2} applies to $(x_j, y_j)$; set $I$ covers the other two cases.
For any $i \in I$, either $x_i = y_i$ or $(x_i, y_i) \in M_1 \cup \overline{M}$.
For any $j \in J$, $(x_j, y_j) \in M_2$ and hence $(y_j, x_j) \in M_1$.  $\STP$
maps $\tup{x}, \tup{y}$ to $\joinIJ{x}{y}, \joinIJ{y}{x}$, so we have
$\gamma(\tup{x}) + \gamma(\tup{y}) < \gamma(\joinIJ{x}{y}) +
\gamma(\joinIJ{y}{x})$. By Lemma~\ref{lmEdgeRAry}, graph $G_\Gamma$ contains edge
$(x_i, y_i) - (y_j, x_j)$ for some $i \in I, j \in J$, which contradicts Lemma~\ref{lmGraphProperties}.
\end{proof}

The following definition corresponds to the ``$\mu$ function''
from~\cite[Section~6]{kz13:jacm}. 

\begin{definition}
For any $a, b, c \in D$, we say that $ab|c$ holds if $a, b, c$ are all different labels
and there exist $(s, t) \in \overline{M}$ such that binary relation $\{ (a, s),
(b, s), (c, t) \}$ belongs to $\clGamma$.
\end{definition}

The intuition is that if $ab|c$ holds, then any minority operation on
$\overline{M}$ must map any permutation of $\{a, b, c\}$ to $c$ in order to be a
polymorphism of $\Gamma$.

\begin{lemma}
\label{lmMuProperties}
For any $a, b, c \in D$, at most one of $ab|c$, $ca|b$, $bc|a$ holds. If $ab|c$,
then $(a, c), (b, c) \in \overline{M}$.
\end{lemma}

\begin{proof}
Suppose that both $ca|b$ and $bc|a$ hold. Then there are $(s_1, t_1), (s_2, t_2)
\in \overline{M}$ and binary relations $\rho_1, \rho_2 \in \clGamma$ such that
$\rho_1 = \{ (c, s_1), (a, s_1), (b, t_1) \}$, $\rho_2 = \{ (b, s_2), (c, s_2),
(a, t_2) \}$. We construct their join $\rho$ as $\rho(x, y) = \min_{z \in
D} (\rho_1(z, x) + \rho_2(z, y))$. We have $\rho \in \clGamma$ and $\rho = \{
(s_1, s_2), (s_1, t_2), (t_1, s_2) \}$, which implies a soft edge in
$\overline{M}$ and hence a contradiction.

If $ab|c$, then there are $(s, t) \in \overline{M}$ such that $\{ (a, s), (b,
s), (c, t) \} \in \clGamma$. By restricting this relation at the first
coordinate to labels $\{a, c\}$ we get edge $(a, c) - (t, s)$ and thus $(a, c)
\in \overline{M}$; analogously by restricting to $\{b, c\}$ we get $(b, c) \in
\overline{M}$.
\end{proof}

We define multimorphism $\MJN$ as follows:
\begin{subnumcases}{\MJN(x, y, z) =}
\label{eqMJN_z}
(x,y,z) & \text{if $x=y\wedge(y,z)\in\overline{M}$ or $xy|z$,} \\
\label{eqMJN_y}
(z,x,y) & \text{if $z=x\wedge(x,y)\in\overline{M}$ or $zx|y$,} \\
\label{eqMJN_x}
(y,z,x) & \text{if $y=z\wedge(z,x)\in\overline{M}$ or $yz|x$,} \\
\label{eqMJN_otherwise}
(x,y,z) & \text{otherwise.}
\end{subnumcases}
Note that the operations of $\MJN$ are majorities and a minority on
$\overline{M}$. Also note that in the subcase $x=y\wedge(y,z)\in\overline{M}$ of
case~\eqref{eqMJN_z}, the output has to be $(x,y,z)$ for $\MJN$ to be an MJN
multimorphism of $\Gamma$ on $\overline{M}$ (and similarly for the first subcase of case~\eqref{eqMJN_y}
and case~\eqref{eqMJN_x}). It is the other cases where there is some freedom and
where we differ from~\cite{kz13:jacm}.

\begin{theorem}
\label{thmMJN}
$\MJN$ is a multimorphism of $\Gamma$.
\end{theorem}

We will actually prove that \eqref{eqMMorpIneq} in Definition~\ref{def:mult} holds with \emph{equality}. 

\begin{proof}

Suppose for the sake of contradiction this is not true for some $r$-ary weighted
relation $\gamma \in \clGamma$ and $\tup{x}, \tup{y}, \tup{z} \in
\Feas(\gamma)$; we choose $\gamma$ so that it has the minimum arity among such
counterexamples. We denote the $r$-tuples to which $\MJN$ maps $(\tup{x},
\tup{y}, \tup{z})$ by $(\tup{f}, \tup{g}, \tup{h})$.

First we show that case \eqref{eqMJN_y} does not occur. Let $I$ be the set of
coordinates $i$ such that case \eqref{eqMJN_y} applies to $(x_i, y_i, z_i)$ and
let $J$ cover the remaining cases. Suppose that $I$ is non-empty, and note that
$\tup{f}_I = \tup{z}_I, \tup{g}_I = \tup{x}_I, \tup{h}_I = \tup{y}_I$. For every
$i \in I$, it holds $(x_i, y_i), (z_i, y_i) \in \overline{M}$ (directly or by
Lemma~\ref{lmMuProperties}), and either $z_i = x_i$ or $z_ix_i|y_i$.

We claim that $\{x_i, y_i, z_i\} \times \{x_j, y_j, z_j\} \subseteq \PrGamma$
for all $i \in I, j \in J$.
We already have $(x_i, x_j), (y_i, y_j), (z_i, z_j)
\in \PrGamma$. It holds
\begin{align}
(x_i, y_j) \in \PrGamma &\iff (y_i, x_j) \in \PrGamma\,, \\
(z_i, y_j) \in \PrGamma &\iff (y_i, z_j) \in \PrGamma \,,
\end{align}
otherwise there would be a soft edge in $\overline{M}$ (i.e.\ soft edge $(x_i,
y_i) - (y_j, x_j)$ and $(z_i, y_i) - (y_j, z_j)$ respectively).

If $(x_i, y_j), (z_i, y_j) \not\in \PrGamma$, then there are edges $(x_i, y_i) -
(y_j, x_j), (z_i, y_i) - (y_j, z_j)$, and hence $(x_j, y_j), (z_j, y_j) \in
\overline{M}$. Because case \eqref{eqMJN_y} does not apply at coordinate $j$, it
holds $z_j \neq x_j$, and therefore labels $x_j, y_j, z_j$ are all distinct.
But then $(x_i, z_j) \not\in \PrGamma$, otherwise we would get
$\{ (x_i, z_j), (x_i, x_j), (y_i, y_j) \} \in \clGamma$ (obtained by domain
restriction of $\PrGamma$), and thus $z_jx_j|y_j$ would hold. Analogously, we
have $(z_i, x_j) \not\in \PrGamma$. This implies $z_i \neq x_i$, and hence
$z_ix_i|y_i$ holds. By domain restriction of $\PrGamma$ we obtain bijection
relation $\{ (x_i, x_j), (y_i, y_j), (z_i, z_j) \} \in \clGamma$; joining it
with a binary relation showing that $z_ix_i|y_i$ gives us $z_jx_j|y_j$, which is
a contradiction.

If $(x_i, y_j) \in \PrGamma$ and $(z_i, y_j) \not\in \PrGamma$, then we have
$z_i \neq x_i$, $z_ix_i|y_i$, and $(z_j, y_j) \in \overline{M}$. It must also
hold $(x_i, z_j) \not\in \PrGamma$, otherwise there would be a soft edge
incident to vertex $(z_j, y_j)$. But then we have $\{ (x_i, y_j), (y_i, y_j),
(z_i, z_j) \} \in \clGamma$, which implies $x_iy_i|z_i$ and contradicts
Lemma~\ref{lmMuProperties}. The case when $(x_i, y_j) \not\in \PrGamma$ and
$(z_i, y_j) \in \PrGamma$ can be ruled out by an analogous argument.

Therefore, we have $(x_i, y_j), (z_i, y_j) \in \PrGamma$. It must also hold
$(x_i, z_j), (z_i, x_j) \in \PrGamma$, otherwise there would be a soft edge in
$\overline{M}$ (incident to vertex $(x_i, y_i)$ and $(z_i, y_i)$ respectively).
Hence, we have shown that $\{x_i, y_i, z_i\} \times \{x_j, y_j, z_j\} \subseteq
\PrGamma$.

Because $\Feas(\gamma)$ is 2-decomposable by Lemma~\ref{lmGamma2Decomposable},
we have $\joinIJ{u}{v} \in \Feas(\gamma)$ for any $\tup{u}, \tup{v} \in
\{\tup{x}, \tup{y}, \tup{z}\}$. It must hold
\begin{equation}
\label{eqY_FGH}
\gamma(\joinIJ{y}{x}) + \gamma(\joinIJ{y}{y}) + \gamma(\joinIJ{y}{z}) =
\gamma(\joinIJ{y}{f}) + \gamma(\joinIJ{y}{g}) + \gamma(\joinIJ{y}{h}) \,,
\end{equation}
otherwise we would obtain a smaller counterexample by pinning $\gamma$ at every
coordinate $i \in I$ to its respective label $y_i$. This gives $\joinIJ{y}{f},
\joinIJ{y}{g}, \joinIJ{y}{h} \in \Feas(\gamma)$; by an analogous argument we get
$\joinIJ{u}{v} \in \Feas(\gamma)$ for any $\tup{u} \in \{\tup{x}, \tup{y},
\tup{z} \}$ and $\tup{v} \in \{\tup{f}, \tup{g}, \tup{h}\}$. By
Lemma~\ref{lmMBarSwappingEq}, it holds
\begin{align}
\label{eqXG}
\gamma(\joinIJ{x}{x}) + \gamma(\joinIJ{y}{g}) &=
\gamma(\joinIJ{x}{g}) + \gamma(\joinIJ{y}{x}) \,, \\
\label{eqZF}
\gamma(\joinIJ{z}{z}) + \gamma(\joinIJ{y}{f}) &=
\gamma(\joinIJ{z}{f}) + \gamma(\joinIJ{y}{z}) \,.
\end{align}
By adding \eqref{eqY_FGH}, \eqref{eqXG}, and \eqref{eqZF} we get
\begin{equation}
\gamma(\joinIJ{x}{x}) + \gamma(\joinIJ{y}{y}) + \gamma(\joinIJ{z}{z}) =
\gamma(\joinIJ{z}{f}) + \gamma(\joinIJ{x}{g}) + \gamma(\joinIJ{y}{h}) \,,
\end{equation}
and hence \eqref{eqMMorpIneq} holds as equality (note that $\tup{f}_I =
\tup{z}_I, \tup{g}_I = \tup{x}_I, \tup{h}_I = \tup{y}_I$). This is a
contradiction; therefore case \eqref{eqMJN_y} does not apply at any coordinate.

Suppose that case~\eqref{eqMJN_x} applies at some coordinate $i$. $\MJN$ maps
$(\tup{y}, \tup{x}, \tup{z})$ to $(\tup{g}, \tup{f}, \tup{h})$, which gives us
another smallest counterexample to the theorem. However, at coordinate $i$ is
now applied case~\eqref{eqMJN_y}, which was proved impossible.

Finally, we have that only cases~\eqref{eqMJN_z} and~\eqref{eqMJN_otherwise} may
occur in a smallest counterexample. But then $\MJN$ maps $(\tup{x}, \tup{y},
\tup{z})$ to $(\tup{x}, \tup{y}, \tup{z})$, and hence the stated equality holds.
\end{proof}

\section{Conclusions}

We have studied the computational complexity of planar VCSPs. For conservative
valued constraint languages on arbitrary finite domains, we have given a
complete complexity classification. For valued constraint language on Boolean
domains, we have given a necessary condition for tractability. The obvious open
problem is to give a complexity classification of Boolean valued constraint
languages, following a classification of crisp Boolean constraint
languages~\cite{Dvorak15:icalp,Kazda17:soda}. Another line of work is to
consider larger domains in the non-conservative setting. As discussed in
Section~\ref{sec:intro}, this might be difficult given the Four Colour Theorem.
Finally, planar restrictions correspond to forbidding $K_5$ and $K_{3,3}$ as
minors. A possible avenue of research is to consider other forbidden minors in
the incidence graph of the VCSP instance.


\newcommand{\noopsort}[1]{}\newcommand{\Zivny}{\noopsort{ZZ}\v{Z}ivn\'y}

\end{document}